\documentclass[letterpaper,11pt]{article}

\usepackage[utf8]{inputenc}
\usepackage[T1]{fontenc}
\usepackage{microtype}

\usepackage[margin=1in]{geometry}

\usepackage[small]{caption}
\usepackage{subcaption}

\usepackage{nicefrac}
\usepackage{amsmath}
\usepackage{amsfonts}
\usepackage{amssymb}
\usepackage{amsthm}
\usepackage{amstext}
\usepackage{thmtools}

\usepackage{csquotes}
\usepackage{xcolor}
\usepackage{graphicx}

\usepackage{mathtools}
\usepackage{xspace}

\usepackage[ruled]{algorithm2e}
\usepackage{bm}
\usepackage{dsfont}
\usepackage{thm-restate}
\usepackage{enumitem}
\usepackage{comment}
\usepackage{soul}

\usepackage{booktabs}

\usepackage{derivative}

\usepackage{url}
\usepackage[colorlinks]{hyperref}
\hypersetup{citecolor=green!60!black,linkcolor=blue!90!black,urlcolor=blue!90!black,breaklinks}
\usepackage[capitalize,noabbrev,nameinlink]{cleveref}

\usepackage[
 backend=biber,
 style=alphabetic,
 citetracker,
 hyperref=auto,
 maxcitenames=5, 
 sortcites,      
 sorting=nyt,
 maxbibnames=12, 
 date=year,      
 isbn=false,     
 url=false,      
 doi=false,      
 eprint=false,   
]{biblatex}
\newbibmacro{string+doiurlisbn}[1]{%
  \iffieldundef{doi}{%
    \iffieldundef{url}{%
      #1
    }{%
      \href{\thefield{url}}{#1}%
    }%
  }{%
    \href{http://dx.doi.org/\thefield{doi}}{#1}%
  }%
}

\DeclareFieldFormat%
[article,inbook,incollection,inproceedings,patent,thesis,unpublished]
  {title}{\usebibmacro{string+doiurlisbn}{#1}}

\newtheorem{theorem}{Theorem}
\newtheorem{lemma}{Lemma}

\newtheorem{proposition}{Proposition}

\Crefname{proposition}{Proposition}{Propositions}
\crefname{appendix}{Appendix}{Appendices}
\Crefname{appendix}{Appendix}{Appendices}

\newcommand{\opt}{\mathrm{OPT}}
\newcommand{\OPT}{\mathrm{OPT}}
\newcommand{\CP}{\mathrm{CP}}
\newcommand{\LP}{\mathrm{LP}}
\newcommand{\DP}{\mathrm{DP}}
\newcommand{\alg}{\mathrm{ALG}}
\newcommand{\ALG}{\mathrm{ALG}}
\newcommand{\advice}{\mathrm{PRD}}

\newcommand{\hx}{\hat{x}}

\newcommand{\hB}{\widehat{B}}

\newcommand{\auxfun}{g}

\newcommand{\fadv}{a}
\newcommand{\frob}{r}

\usepackage[colorinlistoftodos,prependcaption,textsize=scriptsize]{todonotes}
\usepackage[most]{tcolorbox}
\title{Consistency--Robustness Tradeoffs for Online Bipartite Allocation with Multiple Stages}

\author{Alexander Lindermayr\thanks{Institut für Mathematik, Technische Universität Berlin, Germany.} \and Nicole Megow\thanks{Faculty of Mathematics and Computer Science, University of Bremen, Germany.} \and Lauren Paul\footnotemark[2]}
\date{}

\begin{document}

\maketitle
    
\begin{abstract}
We study learning-augmented online bipartite allocation with multiple stages. 
In the $k$-stage vertex-weighted fractional bipartite matching problem, 
demand vertices arrive in $k$ stages, and the algorithm receives possibly 
inaccurate predictions of the allocation in each stage. 
While tight consistency--robustness tradeoffs were known for the two-stage case, 
no nontrivial tradeoff was known for an arbitrary number of stages.

Our main result is the first consistency--robustness tradeoff for $k$-stage 
vertex-weighted fractional bipartite matching with predictions, for every $k\ge2$. 
Let $R_k=1-(1-1/k)^k$. For every $R\in[0,R_k]$, our algorithm is $R$-robust and 
$C_k(R)$-consistent, where $C_k(R)=k(1-R)^{1/k}+R-(k-1)$.
This simultaneously recovers the known tight two-stage tradeoff and the optimal 
prediction-free $k$-stage competitive guarantee $R_k = C_k(R_k)$, 
while strictly dominating the natural randomized coin-flip baseline between these endpoints.

We also present an algorithm for the classical online setting, 
where demands arrive one by one and the number of demands is unknown in advance. 
It has a consistency ratio of at least $C_\infty(R)=1+R+\ln(1-R)$
for a given robustness $R\in[0,1-1/e]$, improving the best previously 
known tradeoff for this problem. Finally, we extend the framework to 
fractional AdWords and fractional predictions. 

Our algorithms are based on stage-wise convex programs with carefully calibrated vertex-dependent penalties. 
The penalties maintain a dynamic safety reserve for each supply vertex, 
balancing protection against adversarial future arrivals with the ability 
to exploit the predicted allocation.
\end{abstract}

%

\thispagestyle{empty}
\newpage
\setcounter{page}{1}

\section{Introduction}

Online bipartite matching is a central problem in online optimization, with
applications in advertising, resource allocation, ride-hailing, and 
matchmaking~\cite{Vazirani-NewBook2023,Mehta13,devanur2022online,DBLP:journals/sigecom/HuangTW24}. 
In the classical vertex-weighted online model, the supply vertices $j \in S$, which form one side of an initially unknown bipartite graph, are known in advance and have weights $w_j$. 
The demand vertices $i \in D$ on the other side arrive online~\cite{KarpVV90,DBLP:conf/soda/AggarwalGKM11}. Each arriving
demand vertex reveals its incident edges and must be fractionally matched irrevocably.
The goal is to maximize the total weight of allocated supply.
The fundamental algorithms \textsc{Ranking}~\cite{KarpVV90,DBLP:conf/soda/AggarwalGKM11} for integral matchings and \textsc{Balance}~\cite{KalyanasundaramP00,BuchbinderJN07} for fractional matchings
achieve at least a fraction of $1-\nicefrac1e\approx 0.632$ of the optimum total weight, which is best possible in the worst case~\cite{KarpVV90}.
The key difficulty is the decision between exploiting the current arrivals and preserving supply for an unknown future.

In many applications, however, neither arrivals nor decisions occur one request at a time. 
A platform may have a small amount of latency, allowing it to collect several
requests and optimize over them jointly. This motivates the $k$-stage
model, in which demand vertices arrive and reveal their corresponding edges to the supply in $k$ stages $D_1,\ldots,D_k$.
The algorithm can choose a fractional allocation for all demands of a given stage~\cite{DBLP:conf/soda/AggarwalGKM11,DevanurJK13,FengNS24}. 
This model interpolates between the offline setting, where everything is known in advance, and the classical online setting, 
where demands arrive one by one.  
Feng and Niazadeh~\cite{DBLP:journals/mansci/FengN25} gave an optimal
$R_k := 1-(1-\nicefrac1k)^k$-competitive algorithm for $k$-stage vertex-weighted bipartite matching for all $k \geq 2$. 
Thus, the $k$-stage model strictly improves over the online $1-\nicefrac1e$ barrier~\cite{KarpVV90} whenever
the number of stages $k$ is finite, and ``converges''
to classical online bipartite matching and the $1-\nicefrac1e$ guarantee as
$k\to\infty$. 

Modern allocation platforms often use forecasts from historical data, demand models, learned predictors, or expert plans to predict how demands should be matched. 
Such information can be valuable when accurate, but blindly following it can be harmful when wrong, especially when worst-case guarantees matter. 
The \emph{learning-augmented framework} captures this tension through two guarantees~\cite{MitzenmacherV22,LykourisV21}. 
An algorithm is $R$-\emph{robust} if it achieves at least an $R$-fraction of the optimum in the worst case, and $C$-\emph{consistent} if it achieves at least a $C$-fraction of the predicted value in the worst case. 
The goal is not to fully trust or ignore the prediction, but to design algorithms that trade off smoothly between these extremes and achieve a strong \emph{Pareto frontier}.


Learning-augmented online allocation problems have been studied in several models. 
Early work by Mahdian, Nazerzadeh, and Saberi~\cite{DBLP:conf/sigecom/MahdianNS07,DBLP:journals/talg/MahdianNS12} studied AdWords with unreliable estimates, providing a precursor of consistency--robustness guarantees; this line was extended to Display Ads and GAP by Spaeh and Ene~\cite{SpaehEne2023}. 
More recently, Choo, Jin, and Shin~\cite{ChooJS25} studied online vertex-weighted fractional bipartite matching with predictions and improved over the natural randomized coin-flip baseline between following the prediction and running the robust algorithm. 
Most closely related to the $k$-stage setting is the work of Jin and Ma~\cite{JinM22}, who characterized the optimal consistency--robustness tradeoff for online fractional vertex-weighted bipartite matching and fractional AdWords for $k=2$. 
This leaves open the natural question of whether nontrivial tradeoffs are possible for $k$-stage vertex-weighted bipartite allocation with predictions when $k\ge3$.

We answer this question affirmatively. We present learning-augmented algorithms for $k$-stage vertex-weighted fractional bipartite matching and fractional AdWords for all $k\ge2$, as well as for the classical online setting without prior knowledge of the horizon. Our guarantees recover the tight two-stage tradeoff of Jin and Ma~\cite{JinM22} and the optimal prediction-free $k$-stage guarantee of Feng and Niazadeh~\cite{DBLP:journals/mansci/FengN25}; they also improve upon the online tradeoff of Choo, Jin, and Shin~\cite{ChooJS25}.

\subsection{Our Results}

Our main result is a learning-augmented algorithm for $k$-stage 
vertex-weighted fractional bipartite matching 
with a strong non-trivial consistency--robustness tradeoff for all $k \geq 2$.

\begin{theorem}\label{thm:main}
Let $k \geq 2$ be an integer, $R_k = 1 - (1 - \nicefrac1k)^k$, and $C_k(R) = k(1-R)^{1/k} + R - (k-1)$.
For every $R \in [0,R_k]$, there is a learning-augmented algorithm for the $k$-stage vertex-weighted fractional bipartite matching problem with integral predictions that is $R$-robust and $C_k(R)$-consistent.
\end{theorem}

\begin{figure}[tb]
	\centering
	\includegraphics[width=0.95\textwidth]{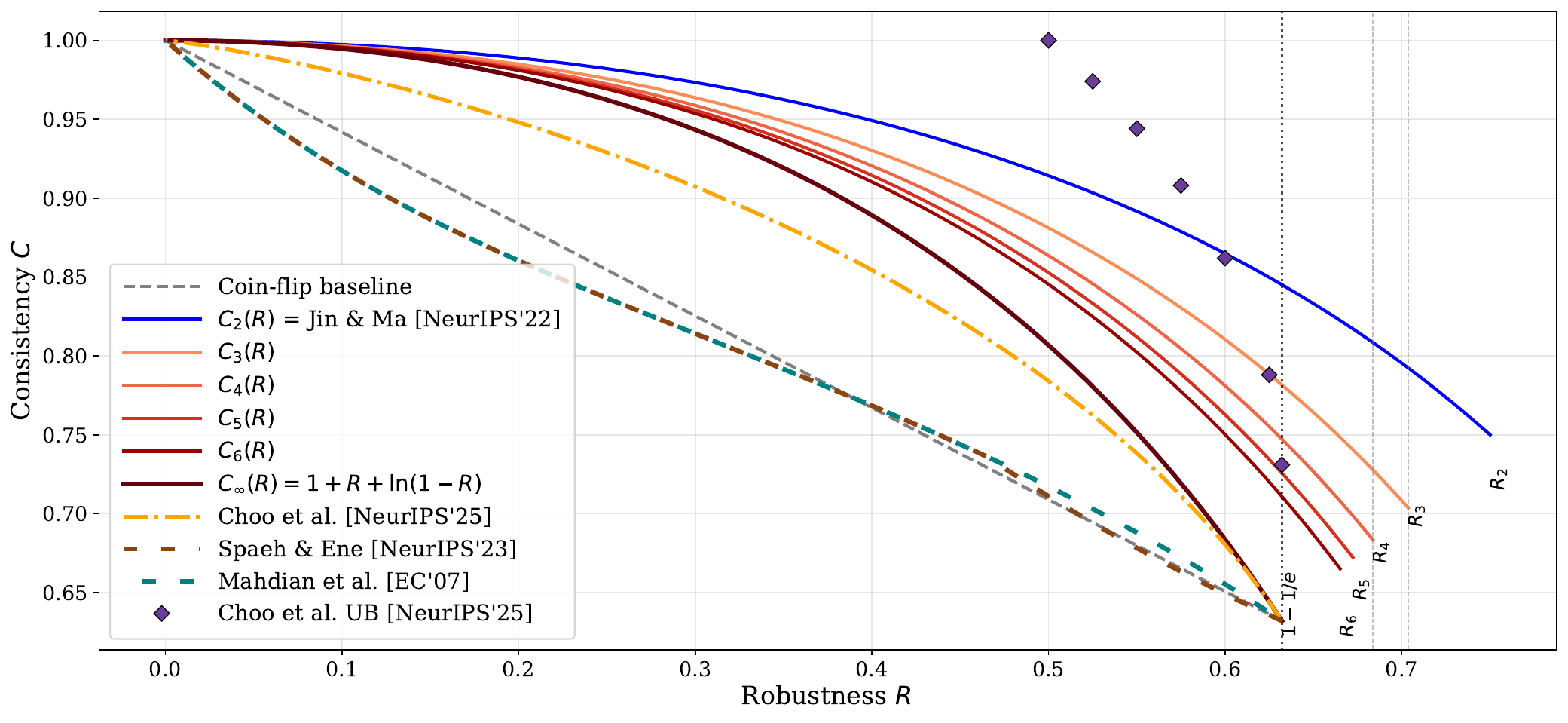}
	\caption{Consistency--robustness tradeoffs. 	Solid curves show our guarantees $C_k$ from \Cref{thm:main,thm:AdWords} and $C_\infty$ from \Cref{thm:online}. Dashed and dotted curves show prior guarantees and upper-bound points.}
	\label{fig:plot_functions}
\end{figure}

Our curve has the right behavior at all benchmark points. 
When $R=0$, we obtain $C_k(0)=1$, corresponding to full 
consistency by completely following the prediction. 
When $R=R_k$, we obtain $C_k(R_k)=R_k$, recovering the 
robust $k$-stage guarantee of Feng and Niazadeh~\cite{DBLP:journals/mansci/FengN25}.
For $k=2$, we have $C_2(R)=2\sqrt{1-R}+R-1$
which is equivalent to the tradeoff $\sqrt{1-R}+\sqrt{1-C_2(R)}=1$ by Jin and Ma~\cite{JinM22}.
Thus, \Cref{thm:main} simultaneously covers the optimal two-stage 
tradeoff and the optimal $k$-stage robust guarantee. For every
$k$, the curve also strictly improves over the naive coin-flip interpolation
between the fully consistent algorithm and the robust
algorithm, which is $(1-\lambda) R_k$-robust and $(R_k + \lambda(1-R_k))$-consistent for
a trust parameter $\lambda \in [0,1]$.
The best-known upper bound on the tradeoff for 
general $k$ is the two-stage upper bound $C_2(R)$ of Jin and Ma~\cite{JinM22}.

The same techniques yield an improved algorithm for the classical online setting, where demands arrive one by one and the number of demands is unknown in advance. 
Viewing the input as a $k$-stage instance with singleton stages and applying \Cref{thm:main} as a black box gives the limiting tradeoff $C_k(R)\xrightarrow{k\to\infty}1+R+\ln(1-R)$, but requires knowing the total number $k$ in advance. 
We therefore take the ``limit of the $k$-stage algorithm'' and reanalyze it to obtain the following theorem.

\begin{theorem}\label{thm:online}
Let $C_\infty(R) := 1 + R + \ln(1-R)$.
For every $R \in [0,1 - \nicefrac1e]$, there is a learning-augmented algorithm for vertex-weighted online fractional bipartite matching with integral predictions that is $R$-robust and $C_\infty(R)$-consistent.
\end{theorem}

\Cref{thm:online} improves over the tradeoff of Choo, Jin, and Shin~\cite{ChooJS25} for online fractional bipartite matching with predictions, who give a $(1+\lambda-\exp(\lambda-1))$-consistent
and $(1-\exp(\lambda-1)-(\exp(\lambda-1) - \lambda) \ln(1-\lambda \exp(1-\lambda)) - \lambda(1-\lambda))$-robust algorithm for every $\lambda \in [0,1]$. 
There are two upper bounds on the tradeoff for the online problem: Jin and Ma show the upper bound $C_2(R)$ via a construction with two demands, hence it also applies here~\cite{JinM22}. Choo et al.\ give a more refined construction for the unweighted online problem, which is the solution to a factor-revealing LP~\cite{ChooJS25}. \Cref{fig:plot_functions} includes representative points from this upper bound construction. Our improvement over~\cite{ChooJS25} is specific to the vertex-weighted and
fractional-AdWords settings.
\citeauthor{ChooJS25}~also provide an algorithm
for the \emph{unweighted} online setting with integral advice, which achieves a
strictly stronger tradeoff in that setting, an integral AdWords guarantee under
the small-bids assumption, and a factor-revealing upper bound for the
unweighted online problem (included in \Cref{fig:plot_functions}).
These results are not subsumed by our framework, and remain the state of the
art in their respective settings.

Finally, we extend our algorithm and results to fractional AdWords,
which is a generalization of vertex-weighted fractional bipartite matching.
Moreover, we consider the more general setting of fractional predictions,
which suggest a fractional allocation in each stage.
Thus, the following theorem is a generalization of \Cref{thm:main} and \Cref{thm:online}
in two different dimensions.

\begin{theorem}\label{thm:AdWords}
Let $k \geq 2$ be an integer, $R_k = 1 - (1 - \nicefrac1k)^k$, and $C_k(R) = k(1-R)^{1/k} + R - (k-1)$.
For every $R \in [0,R_k]$, there is a learning-augmented online algorithm for $k$-stage fractional AdWords with fractional predictions that is $R$-robust and $C_k(R)$-consistent.
Moreover, for every $R \in [0,1 - \nicefrac1e]$, there is a learning-augmented online algorithm for fractional AdWords with fractional predictions that is $R$-robust and $(1+R+\ln(1-R))$-consistent.
\end{theorem}
Our tradeoff curves and those of previous works are visualized in \Cref{fig:plot_functions}.

\subsection{Technical Overview}
\label{sec:technical-overview}

Our algorithm is a stage-wise greedy marginal-value algorithm implemented through convex
programming. In each non-final stage $<k$, assigning load $x$ to a supply vertex $j$ is credited not
with its full marginal value $w_j$, but with the penalized marginal value
$w_j(1-f_j^s(x))$. The penalty $f_j^s(x)$ is the fraction of marginal value converted into
protection for future stages. Large penalties make the allocation conservative and preserve supply
against adversarial future arrivals; small penalties let the algorithm exploit the prediction when it
is safe to do so. In the final stage $k$ there is no future uncertainty left to hedge against, so the algorithm simply solves the residual offline problem.

Our approach combines and extends three lines of work. 
Jin and Ma~\cite{JinM22}
characterize the two-stage problem, where a single anticipatory 
decision is followed by an optimal
residual decision. In that setting, 
a static pair of penalty
functions suffices: lower penalties for stage-1 predicted vertices 
and larger penalties for the
others. Feng and Niazadeh~\cite{DBLP:journals/mansci/FengN25} study the $k$-stage problem without
predictions. Their regularizers form a stage-dependent hedging schedule: the algorithm hedges
more in early stages and becomes greedier as the number of remaining stages decreases. 
Choo, Jin, and Shin~\cite{ChooJS25} use an advice-tracking potential, where the
penalty depends on how the cumulative algorithmic load compares with the cumulative advice load.
These viewpoints identify the right ingredients, but none is sufficient for the 
$k$-stage
prediction setting. 
Static functions do not remember what happened in earlier stages. 
Stage-wise regularizers do not distinguish predicted from 
non-predicted vertices. Advice-tracking potentials
measure whether the algorithm is ahead of or behind the advice, 
but not how much dual protection has already been accumulated 
on a particular vertex. 
Our penalty functions are instead vertex-specific and
history-dependent. 

We analyze the algorithm by dual fitting~\cite{DevanurJK13}. The KKT conditions of the stage convex programs
define dual variables whose objective value equals the value obtained by the algorithm, so the remaining task is approximate dual feasibility. 
This reduction decouples across supply vertices:
for each vertex, the accumulated penalty 
mass together with the remaining capacity and current
marginal value must certify all relevant dual constraints. 
Robustness requires such certificates for
all edges, while consistency requires stronger certificates 
for predicted vertices.

The main new ingredient is the calibration of penalties 
through a dynamic reserve account. 
A safety curve prescribes how much reserve a vertex must 
retain as a function of its current load and
the number of stages still to come. 
This is where the polynomials $g_m(z)=(1-\frac{1-z}{m})^m$ of Feng and Niazadeh~\cite{DBLP:journals/mansci/FengN25} reappear, but in
a different role: in their work they are the marginal regularizers used by the algorithm, whereas in ours they are safety targets for the accumulated dual reserve. 
The penalty charged in a stage is the minimum charge needed to maintain the relevant
certificate. A non-predicted vertex is charged just enough to buy local robustness
protection. A predicted vertex is treated in the opposite direction: the algorithm spends
accumulated reserve to make the predicted edge attractive, but stops exactly
at the safety curve. Thus following the prediction is free whenever the existing reserve
already certifies future safety.
Unrolling this one-dimensional calibration over the history of a vertex and identifying the
worst case yields our tradeoff curve $C_k(R)$. 



\subsection{Further Related Work}

\noindent \textbf{Online matching and allocation problems.}
Online bipartite matching is a central problem in online optimization. The seminal work of Karp, Vazirani, and Vazirani~\cite{KarpVV90} introduced the one-sided vertex-arrival model and established the classical $1-\nicefrac1e$ benchmark for online bipartite matching, together with the tight $\nicefrac12$ bound for deterministic integral algorithms. For vertex-weighted online bipartite matching, Aggarwal, Goel, Karande, and Mehta~\cite{DBLP:conf/soda/AggarwalGKM11} gave a tight randomized $(1-\nicefrac1e)$-competitive algorithm. The same benchmark also appears in fractional matching and allocation through water-filling/\textsc{Balance}-type algorithms~\cite{KalyanasundaramP00,BuchbinderJN07,DevanurJK13}. For AdWords, Mehta, Saberi, Vazirani, and Vazirani~\cite{DBLP:journals/jacm/MehtaSVV07} obtained the analogous tight guarantee in the small-bids regime. 
	We refer to~\cite{Vazirani-NewBook2023,Mehta13,DBLP:journals/sigecom/HuangTW24} for broader overviews. 

\smallskip 
\noindent \textbf{Batched arrivals and other beyond-worst-case models.}
Several models interpolate between fully online and offline optimization by giving the algorithm additional temporal or structural information. Batched-arrival models allow the algorithm to optimize jointly over groups of arrivals; closest to our work is the $k$-stage framework of Feng and Niazadeh~\cite{DBLP:journals/mansci/FengN25}. Related variants include models in which edges rather than vertices arrive in batches~\cite{DBLP:journals/talg/LeeS20}. Other beyond-worst-case approaches include structural assumptions such as degree bounds~\cite{BuchbinderJN07,NaorW18,AlbersS22,CohenP23bounded,CohenW18}, recourse~\cite{AngelopoulosDJ20,MegowN25,BernsteinHR19}, stochastic or random-order arrivals~\cite{DevanurH09,MirrokniGZ12,DBLP:conf/stoc/0002SY22,DBLP:journals/ior/YangY26,FeldmanMMM09,ManshadiGS12}, and sample-based models~\cite{KaplanNR22}. These models are complementary to ours: they improve guarantees by imposing additional assumptions on the input model or by giving reliable auxiliary information, whereas our algorithms use predictions that may~be~inaccurate.

\smallskip 
\noindent \textbf{Predictions in online matching and allocation.}
Online matching has been studied with several forms of imperfect predictions, including predicted degrees~\cite{AamandCI22}, predicted online vertices~\cite{DBLP:conf/sofsem/BurathepEM26,DBLP:conf/icml/ChooGL024,DBLP:journals/ior/YangY26}, predicted matching solutions~\cite{AntoniadisCEPS23,JinM22,KeviN23,BamasMS20,ChooJS25}, and predictions of the optimal value~\cite{AntoniadisGKK23,DBLP:conf/esa/LavastidaM0X21,DBLP:conf/acda/LavastidaM0X21}. Learning-augmented allocation has also been studied for AdWords and related variants~\cite{DBLP:conf/sigecom/MahdianNS07,DBLP:journals/talg/MahdianNS12,SpaehEne2023,VeeVS10,CohenP23}. 

Our guarantees concern vertex-weighted fractional matching and fractional AdWords, not arbitrary edge-weighted matching. Jin and Ma~\cite{JinM22} show that for two-stage edge-weighted matching the naive coin-flip tradeoff is optimal. Whether this remains true for $k\ge3$, or whether additional stages permit a nonlinear improvement, is open. Our analysis does not extend directly to this more general setting.

\smallskip 
\noindent \textbf{Consistency--robustness tradeoffs and smoothness.}
Consistency--robustness tradeoffs are a central theme in learning-augmented algorithms. Tight Pareto frontiers are known only in selected cases, including ski rental~\cite{WeiZ20,DBLP:journals/corr/abs-2312-02547}, sorting~\cite{ErlebachLMS23}, and, closest to our setting, two-stage online matching~\cite{JinM22}. 
Beyond consistency and robustness, much of the 
literature studies smoothness guarantees, which express 
performance as a gradually degrading function of the prediction error~\cite{LykourisV21,GLMMSS22,PurohitSK18,DBLP:conf/soda/AzarPT22}. 
For output predictions, our consistency guarantee can also be interpreted as smoothness against the output-prediction benchmark~\cite{DBLP:conf/aaai/AzarLS25,
DBLP:conf/aaai/EberleLMNS22,DBLP:conf/soda/Gkatzelis0T25}: 
performance is measured relative to the value of the predicted~allocation.

\section{Problem Description}

The input to the \textbf{$k$-stage vertex-weighted fractional bipartite matching} problem is a bipartite graph $G=(S, D, E)$,
where $S$ is a set of \textit{supply vertices}, which are known in
advance, and $D$ is a set of \textit{demand vertices}, which arrive
online in $k$ sequential stages $D_1,\ldots,D_k$. 
The edges $E_s$ incident to demand vertices $D_s$ are also revealed in stage $s$, and define an $s$-th stage graph $G_s = (S, D_s, E_s)$. 
Each $j\in S$ has an associated weight $w_j > 0$. The objective is
to find a maximum-weight fractional matching in~$G$, where each edge incident to supply vertex $j$ has weight $w_j$.
The number of stages $k$ is known in advance.
In \textbf{online vertex-weighted fractional bipartite matching}, the number of stages is unknown in advance, and in each stage only a single demand vertex arrives.

In \textbf{$k$-stage fractional AdWords} each vertex $j \in S$ has a budget $B_j > 0$
and each edge $(i,j) \in E$ has bid $b_{ij} > 0$. 
An allocation $(x_{ij})_{(i,j) \in E}$ is feasible if each $j \in S$ stays within its budget, $\sum_{i: (i,j) \in E} b_{ij}x_{ij} \le B_j$ and each $i \in D$ is matched at most once, $\sum_{j: (i,j) \in E} x_{ij} \le 1$.
The goal is to maximize $\sum_{(i,j)\in E} b_{ij}x_{ij}$.
The vertex-weighted $k$-stage fractional bipartite matching problem is the special case where $b_{ij}=B_j=w_j$.
The online fractional AdWords problem is defined analogously, with one demand vertex arriving at a time and an unknown number of demands.

A \textbf{fractional prediction}
$\hx = (\hx_{ij})_{(i,j) \in E} = 
(\hx^1_{ij} + \ldots + \hx^{k}_{ij})_{(i,j) \in E}$ 
is a feasible fractional solution for the underlying allocation problem. 
An online algorithm receives in each stage $s \in [k-1] := \{1,\ldots,k-1\}$ only the fractional allocation $\hx^s_{ij}$ for edges $(i,j) \in E_s$.
In stage $k$, we assume that the prediction $\hat x^k$
is an optimal residual allocation assuming predicted allocations $\hat x^1,\ldots,\hat x^{k-1}$; $\hat x^k$ is only used in the analysis and the algorithm does not use it.
We write $\advice(G, \hx)$ for the objective value
of $\hx$ for $G$, that is, $\sum_{(i,j)\in E} b_{ij}\hx_{ij}$.
If $\hx$ is integral, we say that it is an \textbf{integral prediction}.
For an integral prediction $\hx$, we denote for each stage $s \in [k]$ by $A_s \subseteq S$ the set of supply vertices that are matched in $\hx$ in stage $s$. 
Since $\hx$ is a matching, the sets $A_1,\ldots,A_k$ are pairwise disjoint.


For an instance $G$, let $\opt(G)$ denote the optimal objective value.
For an algorithm with (integral or fractional) prediction $\hx$, let
$\alg(G, \hx)$ denote the objective value of the solution produced by the algorithm.
An algorithm is $C$-\textbf{consistent} if
\(
    \alg(G, \hx)\ge C\cdot \advice(G, \hx)
\)
for every instance $G$ and every prediction $\hx$.
An algorithm is $R$-\textbf{robust} if
\(
   \alg(G, \hx)\ge R \cdot \opt(G)
\)
for every instance $G$ and every prediction $\hx$.
We only use $\advice$, $\alg$, and $\opt$ if the input is clear from the context.

\section{Algorithmic Framework}\label{sec:algorithm}

We first describe our algorithm for the vertex-weighted $k$-stage bipartite matching problem and integral predictions.
We will describe how it can be applied to the classical online variant in \Cref{sec:online}, and to the fractional AdWords problem and fractional predictions in \Cref{sec:AdWords}.

In each stage $s \in [k] := \{1,\ldots,k\}$, our algorithm computes
a fractional matching $x^s = (x^s_{ij})_{(i,j) \in E_s}$ between the arriving demands $D_s$ and the supply vertices $S$,
such that the cumulative allocation after stage $s$ 
is feasible for all supply vertices and all revealed demand vertices up to that stage.
To this end, let $x_j^{s} := \sum_{(i,j) \in E_s} x_{ij}^s$ denote the total load on supply vertex $j$ in stage $s$, 
$X_j^{s-1} := \sum_{d=1}^{s-1} x_{j}^d$ for the total cumulative load on supply vertex $j$ up to stage $s$ (excluding $s$),
and $x_i^{s} := \sum_{(i,j) \in E_s} x_{ij}^s$ for the total load on demand vertex $i$ in stage $s$.
A feasible allocation for stage $s$ can use at most $1-X_j^{s-1}$ units of supply from vertex $j$.
The objective value of our algorithm is equal to $\alg = \sum_{j \in S} w_j X^{k}_j$.

We use the following convex program $(\CP_s)$ to compute $x^s$ for each stage $s \in [k]$.
\begin{alignat}{3}
		(\CP_s) \quad \max \quad & \sum_{j \in S} w_j \left(x^s_j - \int_0^{x^s_j} f_j^s(t) \, dt\right) \label{eq:convex_program_1} \notag \\
		\text{s.t.} \quad 
		& x_i^s \leq 1 && \forall i \in D_s \notag \\
		& x_j^s \leq 1 - X_j^{s-1} && \forall j \in S \notag \\
		& x^s_{ij} \geq 0 && \forall (i,j) \in E_s \notag 
\end{alignat}
For all non-final stages $s \in [k-1]$, we assign 
a non-decreasing continuous penalty function $f_j^s \colon [0,1-X_j^{s-1}] \to [0,1]$ 
to each supply vertex $j \in S$. 
If
$f_j^s(x)=0$, the algorithm treats $j$ at its full weight $w_j$; if
$f_j^s(x)$ is close to $1$, the algorithm nearly stops using $j$. 
The precise definition of $f_j^s$ is key to proving performance guarantees and depends on whether $j$ is predicted to be matched in the stage $s$. 
We give precise definitions in \Cref{sec:penalty-functions}; until then, it suffices to treat them as a black-box and to use only their properties.
In the final stage $k$, we do not have to hedge against future demand arrivals, and thus,
can simply compute a maximum-weight fractional bipartite matching on the residual instance. Equivalently, we set $f_j^k:=0$ for all $j\in S$. 
We summarize the algorithm as follows.


\begin{algorithm}[H]
\caption{$k$-stage vertex-weighted fractional bipartite matching with predictions}
\DontPrintSemicolon
Initialize $X_j^0:=0$ for every supply vertex $j\in S$ \;
In stage $s$, compute an optimal solution $(x_{ij}^s)_{(i,j)\in E_s}$ to $(\CP_s)$ and augment the allocation chosen in the previous stages by these values. \;
\end{algorithm}

We first record the conditions under which the stage-wise programs $(\CP_s)$ are well-defined. The constraints are linear and explicitly enforce the demand constraints and the residual supply capacities. Hence the feasible region is convex, and every feasible solution extends the current allocation feasibly. It remains to ensure that the objective is concave and that all marginal values are nonnegative.

\begin{lemma}[Convexity]\label{lem:conditional-convexity}
Suppose that, for every $s\in[k]$ and $j\in S$, the penalty function $f_j^s$ is continuous, non-decreasing, and takes values in $[0,1]$. Then $(\CP_s)$ is a well-defined convex optimization problem. Moreover, every feasible solution of $(\CP_s)$ extends the current allocation to a feasible fractional allocation after stage $s$, and all marginal values $w_j(1-f_j^s(x_j^s))$ are nonnegative.
\end{lemma}

\begin{proof}
Each objective term has the form $g_j(x)=w_j (x-\int_0^x f_j^s(t)\,dt)$, with derivative $g_j'(x)=w_j(1-f_j^s(x))$. Since $f_j^s$ is non-decreasing, $g_j$ is concave~\cite{BV2014}.  Moreover, $f_j^s(x)\leq 1$ implies that the marginal value $w_j(1-f_j^s(x))$ is nonnegative, while $f_j^s(x)\geq 0$ ensures that the penalty term never increases the marginal value above the original weight $w_j$. Together with the linear constraints discussed above, this proves the claim.
\end{proof}

Before moving to the analysis, we state a few structural facts about optimal solutions of $(\CP_s)$, based on the KKT conditions. The proof is deferred to \Cref{app:algorithm}.

\begin{restatable}[Preprocessed KKT conditions]{lemma}{lemmaKKT}\label{lem:kkt}
	For each stage $s \in [k]$, there exist non-negative 
	numbers $\lambda^s_i$ for each demand $i \in D_s$ that satisfy the following properties.
	\begin{enumerate}[label=(\alph*),nosep]
		\item For all $(i,j)\in E_s$, if $x^s_{ij} > 0$, then $w_j(1 - f_j^s(x_j^s)) \geq \lambda^{s}_i$. 
		
		\item For all $(i,j) \in E_s$, if there exists $m \in D_s$ with $(m,j) \in E_s$ and $x^s_{mj} > 0$, 
		then $\lambda^{s}_i \geq \lambda^{s}_m$.
	
		\item For all $i \in D_s$, if $x^s_i < 1$, then $\lambda^{s}_i = 0$.
		
		\item For any $(i,j) \in E_s$, if $X^s_j < 1$, then $\lambda^{s}_i \geq w_j (1 - f^s_j(x^s_j))$.

	\end{enumerate}
\end{restatable}

\section{Dual Fitting Analysis for $k$-Stage Bipartite Matching}\label{sec:dual_fitting}

In this section, we perform the main analysis of our algorithm from \Cref{sec:algorithm} for the $k$-stage bipartite matching problem with integral predictions, and prove \Cref{thm:main}.

We use \emph{dual fitting} to compare the objective value of 
our algorithm to an offline optimum~\cite{DevanurJK13}.
The following linear program models an optimal solution to our problem. 
For each edge $(i,j) \in E$, the variable $z_{ij}$ models its fractional allocation. 
Its dual has variables $\alpha_i$ for all $i \in D$ and $\beta_j$ for all $j \in S$. The primal and dual linear programs (LPs) are as follows.

\vspace*{-0.5cm}

\begin{minipage}[t]{0.48\textwidth}
	\begin{alignat}{3}
		(\LP) \quad \max \quad  \sum_{(i,j)\in E}\ & w_j z_{ij} \notag \\
		\text{s.t.}\quad
		\sum_{j:(i,j)\in E} & z_{ij} \leq 1 && \quad \forall i\in D \notag \\
		\sum_{i:(i,j)\in E} & z_{ij} \leq 1 && \quad \forall j\in S \notag \\
		& z_{ij} \geq 0 && \quad \forall (i,j)\in E \notag
	\end{alignat}
\end{minipage}%
\hfill%
\begin{minipage}[t]{0.48\textwidth}
	\begin{alignat}{3}
		(\DP) \quad \min \quad & \sum_{i\in D} \alpha_i + \sum_{j\in S}  \beta_j \notag \\
		\text{s.t.}\quad
		& \alpha_i + \beta_j  \geq w_j &\quad&  \forall (i,j)\in E\notag \\
		& \alpha_i, \beta_j  \geq 0 &&  \forall i\in D,  j\in S \notag
	\end{alignat}
\end{minipage}

\bigskip

For every stage $s\in[k]$, let $(\lambda_i^s)_{i\in D_s}$ be the dual multipliers given by \Cref{lem:kkt}.  
To perform the dual fitting analysis, we next define an allocation of dual variables $(\alpha,\beta)$ of $(\DP)$:
\begin{itemize}[nosep]
\item For every demand vertex
$i\in D_s$, we define
\(
    \alpha_i:=\lambda_i^s .
\)
\item For every supply vertex $j\in S$, we define
\(
    \beta_j
    :=
    \sum_{s=1}^{k}
    (
        w_jx_j^s-\sum_{i\in D_s}\lambda_i^s x_{ij}^s
    ).
\)
\end{itemize}


\subsection{Dual Objective Value}\label{sec:dual-objective}

We first show that the objective value of $(\DP)$ 
for the defined dual solution is equal to the 
objective value $\alg$ of our algorithm, which is equal to its total weighted allocation through all $k$ stages.

The proof of the following lemma mainly relies on \Cref{lem:kkt}(c): if $x_i^s < 1$, then $\alpha_i = 0$. Thus, each $\alpha_i = \lambda_i^s$, if non-zero, absorbs the $- \lambda_i^s x_{ij}^s$ term in the definition of $\beta_j$ since $x_i^s = 1$; see \Cref{app:dual_objective}.

\begin{restatable}[Dual objective value]{lemma}{lemmaDualObjective}
\label{lem:dual_objective}
The objective value $\ALG$ of the algorithm is equal to the objective
value of the allocation $(\alpha,\beta)$ of $(\DP)$, that is, $\ALG=\sum_{i\in D}\alpha_i+\sum_{j\in S}\beta_j$.
\end{restatable}

\subsection{Dual Feasibility}\label{sec:dual-feasibility}

Showing approximate dual feasibility will complete the dual fitting argument: if for some constant $\mu \in (0,1]$, for every edge $(i,j) \in E$ holds
$\alpha_i + \beta_j \geq \mu w_j$,
then the scaled dual solution $(\alpha / \mu, \beta / \mu)$ 
is feasible for $(\DP)$ and has by \Cref{lem:dual_objective} an objective value equal to
\[
    \frac{\alg}{\mu} =  \frac{1}{\mu} \biggl( \sum_{i \in D} \alpha_i + \sum_{j \in S} \beta_j \biggr)
    \geq \opt \ ,
\]
where the last inequality holds by weak LP duality between $(\LP)$ and $(\DP)$ and since $(\LP)$ models an optimal solution. This implies that the algorithm is $\mu$-competitive.

In the following, we will apply this argument twice, once for robustness with $\mu = R$ and once for consistency with $\mu = C_k(R)$, where in the latter we compare against the predicted solution's objective value $\advice$ instead of an optimum by restricting the instance to the predicted allocation.
We first introduce more
shorthand notation.
\begin{itemize}
	\item We set $P_j^s:=\sum_{d=1}^{s} x_j^d \cdot f_j^d(x_j^d)$ for all $j \in S$ and all $s \in [k]$. We define $P_j^0:=0$.  
	\item For each $j \in S$, let $q_j$ denote the first stage $q \in [k]$ where $X_j^q = 1$, that is, $j$ becomes fully saturated; if $j$ never becomes fully matched we set $q_j = \infty$.
\end{itemize}
%
%
We first use \Cref{lem:kkt} to derive intermediate lower bounds on the dual constraint $\alpha_i + \beta_j$ for our solution. The proof is deferred to \Cref{app:dual_feasibility}. We use them for both consistency and robustness. 

\begin{restatable}[Intermediate Lower Bounds]{lemma}{lemmaDualConstraintLowerBound}\label{lem:intermed_feasibility}
For every $j\in S$, it holds that $\beta_j\ge w_jP_j^{k-1}$.
Moreover, for every edge $(i,j)\in E_s$ with $s\in[k]$, one of the following cases applies:
\begin{enumerate}[label=\arabic*.,nosep]
    \item If $X_j^s<1$, then $\alpha_i+\beta_j \geq w_j(P_j^{k-1}+1-f_j^s(x_j^s))$.

    \item If $q_j<s$, then $\alpha_i+\beta_j \ge w_jP_j^{q_j}$.

    \item If $q_j=s$, then $\alpha_i+\beta_j \geq w_j(P_j^{s-1}+x_j^s)$.
\end{enumerate}
In particular, for every final-stage edge $(i,j)\in E_k$,
we have $\alpha_i+\beta_j \geq w_j(P_j^{k-1}+1-X_j^{k-1})$.
\end{restatable}


We next use \Cref{lem:intermed_feasibility} to carry out the dual-fitting argument for consistency and robustness, based on suitable load certificates for the penalty functions. 
The certificates, namely \eqref{eq:k-rob-s} and \eqref{eq:k-rob-k} for robustness and \eqref{eq:k-cons-s} and \eqref{eq:k-cons-prefix} for consistency, are verified for our concrete penalties in the following subsections. 
The next two lemmas combine \Cref{lem:intermed_feasibility} with these certificates and provide the conditional dual-fitting guarantees to prove \Cref{thm:main}; their proofs are deferred to~\Cref{app:dual_feasibility}.

Before stating the lemmas, let us briefly interpret the quantities that appear in their assumptions.  These quantities have a clean dual-fitting meaning.
The term $P_j^{k-1} = P_j^k$
is the total dual mass accumulated on $j$ through penalties before the final
stage. The expression $P_j^{k-1}+1-f_j^s(x_j^s)$
is the normalized (by $w_j$) dual mass available for an edge incident to $j$ in stage
$s$ when $j$ is not yet saturated: the first term comes from accumulated
penalties and the second from the current marginal value. The expression
$P_j^{k-1}+1-X_j^{k-1}$
is the corresponding final-stage quantity: accumulated penalties plus the
remaining capacity of $j$. Finally, $P_j^p+1-X_j^p$
measures whether a vertex predicted for a future stage $s$ still has enough
``penalty plus residual capacity'' after an earlier prefix of stages~$p < s$. 

The robustness proof requires these certificates to be at least $R$, while the
consistency proof requires the same type of certificate to be at least
$C_k(R)$ on predicted vertices.
Thus, after the dual-fitting reduction, the remaining task is a one-dimensional load-calibration problem for each supply vertex,
guided by a safety target introduced in the next section.

\begin{restatable}[Conditional robustness]{lemma}{lemmaConditionalRobustness}
    \label{lem:conditional-rob}
	If for all stages $s \in [k-1]$ and supply vertices $j \in S$ there exist continuous non-decreasing penalty functions $f_j^s \colon [0,1 - X_j^{s-1}] \rightarrow [0,1]$ such that 
	\begin{equation}
		P_j^{k-1} + 1 - f_j^s(x_j^s) \geq R  \label{eq:k-rob-s} \tag{$R_ks$}
	\end{equation}
	for all $s \in [k-1]$ and $j \in S$, 
    and
	\begin{equation}
		P_j^{k-1} + 1 - X_j^{k-1} \geq R\label{eq:k-rob-k} \tag{$R_k k$}
	\end{equation}
	for all $j \in S$,
	then the algorithm is $R$-robust.
\end{restatable}

\begin{restatable}[Conditional consistency]{lemma}{lemmaConditionalConsistency}
    \label{lem:conditional-cons}
	If for all stages $s \in [k-1]$ and supply vertices $j \in S$ there exist continuous non-decreasing penalty functions $f^s_j \colon [0,1 - X_j^{s-1}] \rightarrow [0,1]$ such that
        \begin{equation}
            P_j^{k-1} + 1 - f_j^s(x_j^s) \geq C_k(R)  \label{eq:k-cons-s}  \tag{$C_k s$}
        \end{equation}
        for every $s\in[k-1]$ and every $j\in A_s$,
        and
        \begin{equation}
            P_j^{p} + 1 - X_j^{p} \geq C_k(R) \label{eq:k-cons-prefix}  \tag{$C_k \mathrm{prefix}$} 
        \end{equation}
    for every $s \in [k]$, every $p \in[s-1]$ and every $j\in A_s$,
	then the algorithm is $C_k(R)$-consistent.
\end{restatable}

\subsection{Definition and Properties of Penalty Functions}\label{sec:penalty-functions}

It remains to define penalty functions that satisfy the conditions of \Cref{lem:conditional-rob,lem:conditional-cons} for the allocation produced by the algorithm.
The following function $g_m$ is the \enquote{safety curve} for the remaining stages. 
We use it as a target in the invariant
\(
    1-R+P_j^s\ge \auxfun_{k-s}(X_j^s).
\)
Thus, if there are $m$ stages still to account for and a vertex has cumulative
load $z$, then $g_m(z)$ specifies the lower envelope that the normalized quantity
$1-R+P_j^s$ should dominate:
	\[
	\auxfun_m(z)
	:= \left(1-\frac{1-z}{m}\right)^m \ .
	\]	
Now fix a supply vertex $j \in S$.
We inductively define penalty functions $f_j^s$ for all $s \in [k-1]$.
The notation suppresses their dependence on the robustness parameter $R$ and on the history variables $X_j^{s-1}$ and $P_j^{s-1}$.
Thus, for $s \in [k-1]$, given penalty functions $f_j^1,\ldots,f_j^{s-1}$ and already fixed loads $x_j^1,\ldots,x_j^{s-1}$, we define the \emph{non-predicted penalty function} 
for stage $s$ by setting for all $z \in [0,1)$
	\[
		r^s_j(z)
		:= \min\left\{
		1,
		\frac{1-R+ P^{s-1}_j}{1-z}
		\right\}
	\]
	and continuously extend to $z=1$, which gives $r^s_j(1) = 1$.
    For the predicted penalty function for stage $s$, 
    define for all $z \in (0,1 - X_j^{s-1}]$ 
	\[
		a_j^s(z)
		:= \max\left\{
		0,
		\frac{\auxfun_{k-s}(z + X_j^{s-1}) - (1 - R + P^{s-1}_j)}{z} 
		\right\} ,
	\]
	and continuously extend to $z=0$, which gives $a_j^s(0) = 0$.
    To see the extension explicitly, let $m=k-s$ and $X=X_j^{s-1}$. If $X<1$, the invariant proved below gives
    $1-R+P_j^{s-1}\ge g_{m+1}(X)>g_m(X)$. Here, the strict inequality follows from the strict AM-GM step in the proof of \Cref{lem:g-props}(b). Hence the numerator is negative for all sufficiently small positive $z$, so $a_j^s(z)=0$ near the origin. If $X=1$, the feasible domain is the singleton $\{0\}$, on which we set $a_j^s(0)=0$.
    
    Finally, for all $z \in [0, 1- X_j^{s-1}]$, define
	\[
	f^s_j(z) := \begin{cases}
		a_j^s(z) & \text{if } j \in A_s  , \\
		r_j^s(z) & \text{otherwise}.
	\end{cases}
	\]
The two branches of $f_j^s$ are calibrated against the same future-safety target, but in different ways. 
For a non-predicted vertex $j\notin A_s$, the algorithm protects against arbitrary future behavior by making the local robustness certificate tight. 
Ignoring future penalty contributions, the stage-$s$ robustness inequality would be $P_j^s+1-r_j^s(x_j^s)\ge R$. 
Since $P_j^s=P_j^{s-1}+x_j^s r_j^s(x_j^s)$, making this inequality tight gives
\[
    r_j^s(x_j^s)=\frac{1-R+P_j^{s-1}}{1-x_j^s},
\]
with the cap at $1$ ensuring nonnegative marginal values. 
For a predicted vertex $j\in A_s$, the algorithm can allocate more aggressively. 
Instead of forcing the current robustness inequality to be tight, it uses the safety target directly and solves the equality version of $1-R+P_j^s\ge g_{k-s}(X_j^s)$ for the current penalty. 
This gives the formula for $a_j^s$. 
The maximum with $0$ says that, if the invariant is already satisfied, no additional penalty is needed. 
\Cref{fig:stage2-penalties} visualizes these penalties for $k=3$.

\begin{figure}[tb]
\centering
\includegraphics[width=0.32\textwidth]{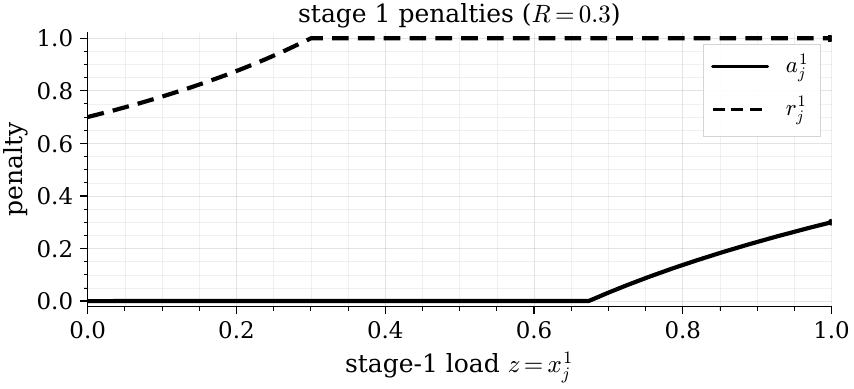}
\hfill
\includegraphics[width=0.32\textwidth]{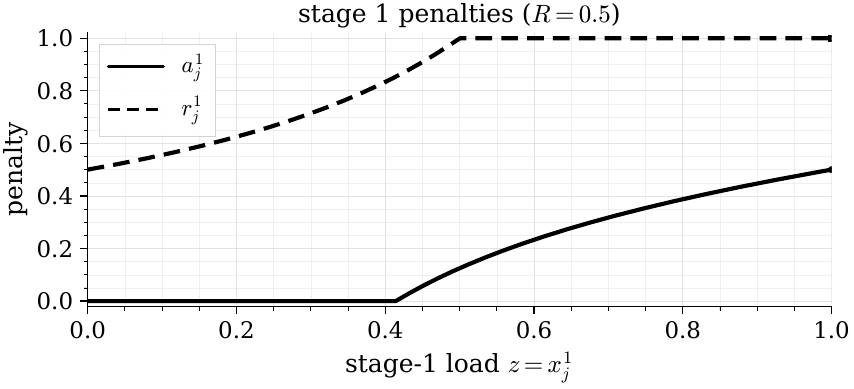}
\hfill
\includegraphics[width=0.32\textwidth]{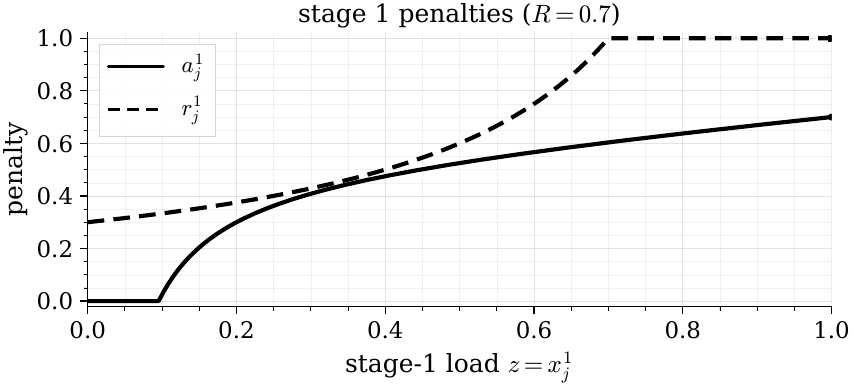}

\medskip
    \includegraphics[width=0.32\textwidth]{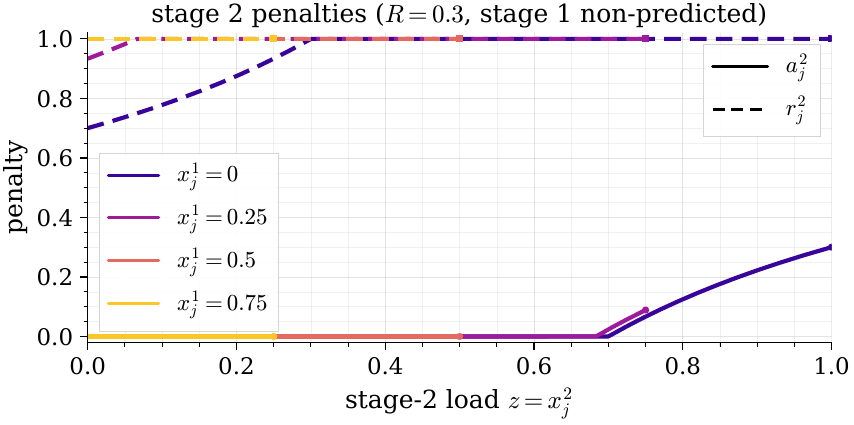}
    \hfill
    \includegraphics[width=0.32\textwidth]{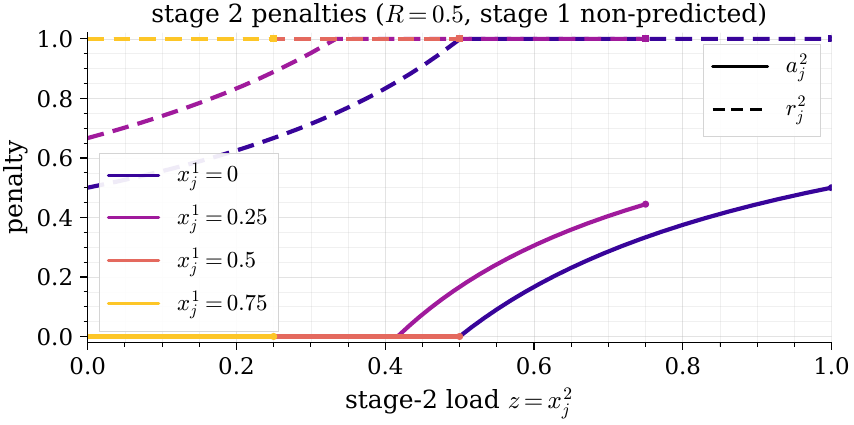}
    \hfill
    \includegraphics[width=0.32\textwidth]{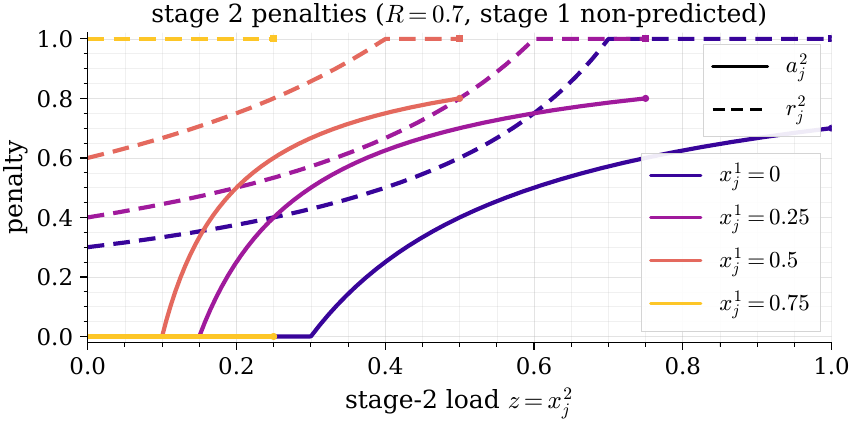}

    \medskip
    \includegraphics[width=0.32\textwidth]{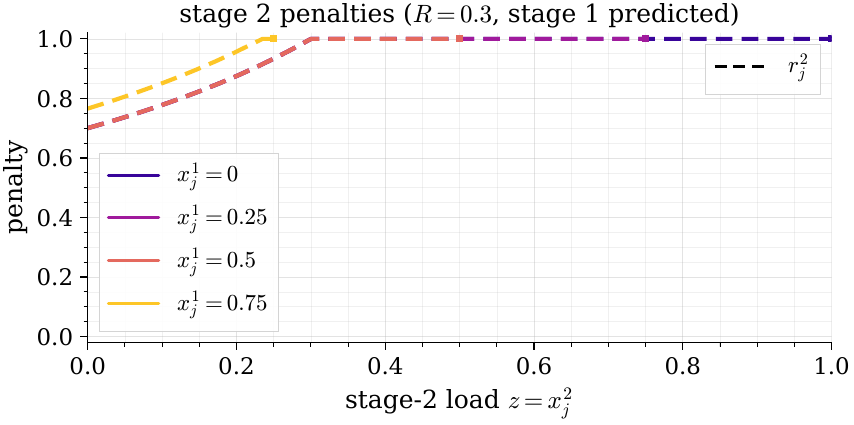}
    \hfill
    \includegraphics[width=0.32\textwidth]{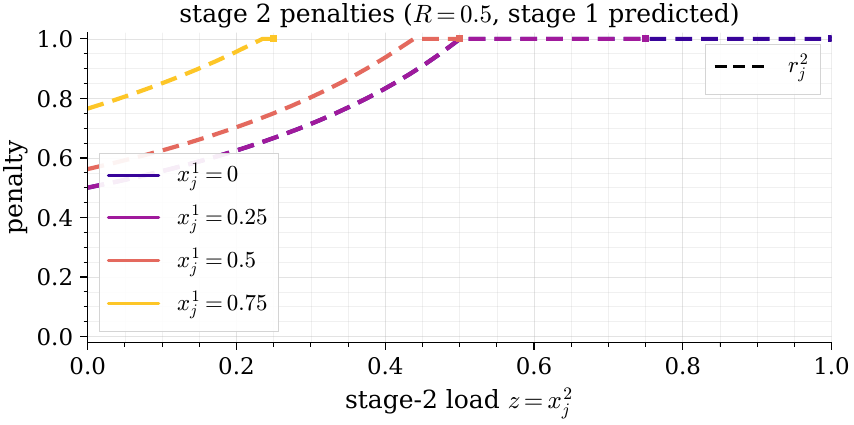}
    \hfill
    \includegraphics[width=0.32\textwidth]{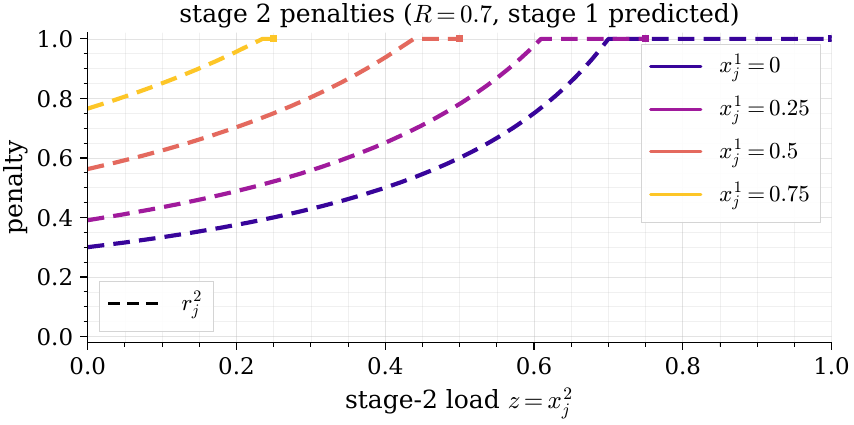}
    \caption{
    \textbf{First row:} stage-1 penalty functions for $k=3$ across robustness levels~$R$.
    \textbf{Second row:} stage-2 penalty functions for $k=3$ across robustness levels assuming that the vertex was not predicted in stage~1.
    \textbf{Third row:} stage-2 penalty functions for $k=3$ across robustness levels assuming that the vertex was predicted in stage~1 and hence $a_j^2$ will not be used in stage 2.}
    \label{fig:stage2-penalties}
\end{figure}



The safety curve $\auxfun_m(z)$ will be particularly useful in combination with the well-known inequality of the arithmetic and geometric means (AM-GM), 
which we will rely on extensively. 

\begin{proposition}[AM-GM]
\label{am-gm}
	Let $z_1,\ldots,z_n$ be non-negative real numbers. Then their geometric mean is at most their arithmetic mean, that is,
	$(z_1 \cdot \ldots \cdot z_n)^{1/n} \leq \frac1n (z_1 + \ldots + z_n)$.
\end{proposition}

We first use AM-GM to prove the following properties in \Cref{app:properties-penalties}.

\begin{restatable}{lemma}{lemmaLambdaProps}\label{lem:g-props}
The following statements are true.
\begin{enumerate}[label=(\alph*),nosep]
	\item For every $s \in \{1,\ldots,k\}$ and $z\in[0,1]$, it holds that $1 \geq \auxfun_s(z) \geq z$.
	\item For every $s \in \{2,\ldots,k\}$, $z \in [0,1]$, and $y \in [0,1-z]$, it holds that
	\(
	\auxfun_s(z) \geq (1-y) \cdot \auxfun_{s-1}(z+y).
	\)
	\item For every $s \in \{2,\ldots,k\}$, $z \in [0,1]$, and $y \in [0,1-z]$, it holds that
	\(
	\auxfun_s(z)+y \geq \auxfun_{s-1}(z+y).
	\)
\end{enumerate}
\end{restatable}
Parts (b) and (c) of \Cref{lem:g-props} are precisely the two update rules needed in the
induction. 
Part (b) handles the uncapped non-predicted $r_j^s$, where the quantity
$1-R+P_j^{s-1}$ is divided by the remaining current-stage capacity $1-z$. 
Part (c) handles the capped or additive case, where the update pays directly
through the current load.


%
%

The next lemma verifies the assumptions of \Cref{lem:conditional-convexity} for our penalty functions~(\Cref{app:validity-penalties}).

\begin{restatable}{lemma}{LemmaPenaltiesValid}
\label{lem:penalties-valid}
For every stage $s\in[k]$ and every supply vertex $j\in S$, the penalty function $f_j^s$ is continuous, non-decreasing, and takes values in $[0,1]$.
\end{restatable}

\subsection{Proof of Consistency and Robustness Certificates and \Cref{thm:main}}
\label{sec:proofs-main-theorem}



We finally prove the sufficient consistency and robustness certificates for our penalty functions, which we stated in \Cref{lem:conditional-rob} and \Cref{lem:conditional-cons}.
Before we start, recall that $R_k = 1 - (1 - \nicefrac{1}{k})^k$ and $C_k(R) = k(1-R)^{1/k} + R - (k-1)$.
We use that $R \leq C_k(R)\leq 1$, which is easy to verify.

We start with the consistency certificate \eqref{eq:k-cons-prefix}. Its proof relies on the following auxiliary lemma, which characterizes the weighted penalty $P_j^h$ if $j$ was not predicted in previous stages $[h-1]$
and its penalty function was never capped, that is, $r_j^d(x_j^d)<1$ for every $d\in [h]$.
Then we can inductively substitute  $r_j^d$ and obtain the following closed-form expression (proof in \Cref{app:proofs-conditions}).

\begin{restatable}{lemma}{lemmaUncappedNonpredicted}\label{claim:uncapped-nonpredicted-product}
Fix a supply vertex $j \in S$ and an index $h\in\{0,1,\ldots,k-1\}$ 
such that $j\notin A_d$ for every $d\in [h]$. 
If $r_j^d(x_j^d)<1$ for every $d\in [h]$, then
\[
    1-R+P_j^h
    =
    \frac{1-R}{\prod_{\ell=1}^h(1-x_j^\ell)} .
\]
Moreover, if $t$ is the first index in $[h]$ such that $r_j^t(x_j^t)=1$, 
then the following two conclusions hold:
\begin{enumerate}[label=(\roman*),nosep]
    \item $r_j^d(x_j^d)=1$ for every $d\in\{t,\ldots,h\}$, and
    \item consequently, $P_j^d-X_j^d=P_j^{t-1}-X_j^{t-1}$ for every $d \in \{t,\ldots,h\}$.
\end{enumerate}
\end{restatable}

\begin{lemma}\label{Lemma:feasibility_consistency1}
The penalty functions satisfy \eqref{eq:k-cons-prefix}, that is, 
for every $s\in[k]$, every $j\in A_s$, and every $p\in [s-1]$, it holds that
    \(
        P_j^p+1-X_j^p \geq C_k(R).
    \)
\end{lemma}

\begin{proof}
Fix a stage $s\in[k]$, a supply vertex $j\in A_s$, and a prefix stage $p\in [s-1]$.
Since vertex $j$ is not predicted in any stage $d\leq p \leq s-1$, we have $f_j^d=r_j^d$.
Let $t$ be the first index in $[p]$ such that $r_j^t(x_j^t)=1$; if no such index exists, set $t=p+1$.
If $t=p+1$, then \Cref{claim:uncapped-nonpredicted-product} with $h=p$ gives
\[
    1-R+P_j^{t-1}
    =
    \frac{1-R}{\prod_{\ell=1}^{t-1}(1-x_j^\ell)}
\]
and the identity $P_j^p-X_j^p=P_j^{t-1}-X_j^{t-1}$ is trivial. If $t\leq p$, then \Cref{claim:uncapped-nonpredicted-product} with $h=t-1$ gives the same product formula, and the second part of the claim, applied with $h=p$, gives
\(
    P_j^p-X_j^p=P_j^{t-1}-X_j^{t-1}.
\)
Using these identities, we obtain
\begin{align*}
    P_j^p+1-X_j^p+(k-1)-R
    &= P_j^{t-1}+1-X_j^{t-1}+(k-1)-R \\
    &= \frac{1-R}{\prod_{\ell=1}^{t-1}(1-x_j^\ell)} - X_j^{t-1}+k-1 \\
    &= \sum_{\ell=1}^{t-1}(1-x_j^\ell) + \sum_{\ell=1}^{k-t}1 + \frac{1-R}{\prod_{\ell=1}^{t-1}(1-x_j^\ell)}.
\end{align*}
More explicitly, the last expression consists of the $t-1$ numbers
$1-x_j^1,\ldots,1-x_j^{t-1}$, followed by $k-t$ copies of $1$, and the final number
$(1-R)/\prod_{\ell=1}^{t-1}(1-x_j^\ell)$. Thus there are
$(t-1)+(k-t)+1=k$ nonnegative numbers, and their product is $1-R$.
By AM-GM their sum is at least $k(1-R)^{1/k}$. We conclude by rearranging
\[
    P_j^p+1-X_j^p \geq k(1-R)^{1/k}+R-(k-1)=C_k(R). \qedhere
\]
\end{proof}

The proof of property \eqref{eq:k-cons-s} in the lemma below is deferred to \Cref{app:proofs-conditions}.

\begin{restatable}{lemma}{lemmaFeasibilityConsistencyTwo}\label{Lemma:feasibility_consistency2}
The penalty functions satisfy \eqref{eq:k-cons-s}, that is, for every $s\in[k-1]$ and every $j\in A_s$, it holds that
    \(
        P_j^{k-1}+1-f_j^s(x_j^s) \geq C_k(R) .  
    \)
\end{restatable}



We next prove the sufficient robustness certificate of our penalty functions.

\begin{lemma}\label{Lemma:feasibility_robustness1}
The penalty functions satisfy \eqref{eq:k-rob-s} for every
$s\in [k-1]$.
\end{lemma}

\begin{proof}
We first prove \eqref{eq:k-rob-s}. Fix $s\in[k-1]$ and $j\in S$.
Recall that $C_k(R)\geq R$.
If $j\in A_s$, then \eqref{eq:k-cons-s} from \Cref{Lemma:feasibility_consistency2} immediately gives
$P_j^{k-1}+1-f_j^s(x_j^s) \geq C_k(R) \geq R$.
It remains to consider the case $j \notin A_s$. Here we have  $f_j^s=r_j^s$.
We claim that
\begin{equation}\label{eq:robust-stage-local}
    1+P_j^{s-1}+(x_j^s-1) \cdot r_j^s(x_j^s) \geq R.
\end{equation}
If $x_j^s=1$, then $r_j^s(x_j^s)=1$, and \eqref{eq:robust-stage-local} becomes
$1+P_j^{s-1}\geq R$, which is true since $R\leq 1$ and $P_j^{s-1}\geq 0$.
Now assume $x_j^s<1$. If $r_j^s(x_j^s)=1$, then the definition of $r_j^s$ gives
\[
    \frac{1-R+P_j^{s-1}}{1-x_j^s}\geq 1,
\]
which is equivalent to
\(
    1+P_j^{s-1}+(x_j^s-1) \geq R.
\)
This is exactly \eqref{eq:robust-stage-local} if $r_j^s(x_j^s)=1$.
If $r_j^s(x_j^s)<1$, that is, $ r_j^s(x_j^s)$ is uncapped,
then $1+P_j^{s-1}+(x_j^s-1)r_j^s(x_j^s)=R$.
This proves \eqref{eq:robust-stage-local} in all cases.

Using \eqref{eq:robust-stage-local}, we conclude
\begin{align*}
    P_j^{k-1}+1-f_j^s(x_j^s)
    &= P_j^{k-1}+1-r_j^s(x_j^s) \\
    &\geq P_j^s+1-r_j^s(x_j^s) 
    = 1+P_j^{s-1}+(x_j^s-1) \cdot r_j^s(x_j^s) 
    \geq R.
\end{align*}
Thus \eqref{eq:k-rob-s} holds.
\end{proof}

Finally, the proof of \eqref{eq:k-rob-k} in the lemma below is deferred to \Cref{app:proofs-conditions}.

\begin{restatable}{lemma}{lemmaRobustnessTwo}\label{Lemma:feasibility_robustness2}
	The penalty functions satisfy \eqref{eq:k-rob-k}.
\end{restatable}

%
\begin{proof}[Proof of \Cref{thm:main}]
By \Cref{lem:conditional-convexity,lem:penalties-valid}, the stage-wise convex programs are well-defined and compute solutions that are feasible with the matching decisions taken previously.

\Cref{Lemma:feasibility_consistency1,Lemma:feasibility_consistency2} show that our penalty functions satisfy the certificates of \Cref{lem:conditional-cons}; hence \Cref{lem:conditional-cons} implies that the algorithm is $C_k(R)$-consistent.
\Cref{Lemma:feasibility_robustness1,Lemma:feasibility_robustness2} show that our penalty functions satisfy the certificates of \Cref{lem:conditional-rob}, which implies that the algorithm is $R$-robust.
\end{proof}


\section{Online Fractional Bipartite Matching with Predictions}\label{sec:online}

We now describe our algorithm for the classical online setting, where demands arrive one by one and the number of demands is not known in advance. The goal is to prove \Cref{thm:online}. 

An online input with one demand vertex at a time can of course be viewed as a $k$-stage instance with singleton stages, where $k$ is the total number of arrivals. This observation already gives a useful corollary of \Cref{thm:main} assuming that the number of demands is known in advance. It does not, however, immediately give the online result in \Cref{thm:online}. The reason is that our $k$-stage algorithm of \Cref{sec:algorithm} uses the value of $k$ inside the penalties through the safety curve
$\auxfun_{k-s}(z)=(1-\frac{1-z}{k-s})^{k-s}$,
and it also treats the last stage differently by setting the penalty to zero. 
The purpose of this section is to identify penalty 
functions that are independent of the number of remaining stages
and to explain why the same proof strategy still works.

The key observation is that the $k$-stage safety curves have the stationary limit
\[
\auxfun_m(z)= \left(1-\frac{1-z}{m} \right)^m
    \xrightarrow{m \to \infty}
    \exp(z-1) =: \auxfun_\infty(z).
\]
In the $k$-stage analysis, $\auxfun_m(z)$ is the safety target for the
normalized quantity $1-R+P_j^s$ after cumulative load $z$ when there are $m$
future stages to protect against. In the online model there is no known final
stage where the algorithm can stop hedging; we thus replace it with
$\auxfun_\infty(z)=\exp(z-1)$.

We use the same notation as for the $k$-stage setting.
Demand vertices arrive one by one, equivalently $|D_t|=1$.  When $D_t$ arrives, the algorithm learns the incident edges and the part of the integral predicted matching incident to $D_t$.  Let $A_t\subseteq S$ be the set of supply vertices used by the prediction at time $t$; since the prediction is a feasible matching, the sets $A_t$ are pairwise disjoint. 
Let $T$ denote the length of the online input sequence, which is unknown to the algorithm.

Fix $R\in[0,1-\nicefrac1e]$ and define $C_\infty(R):=1+R+\ln(1-R)$.
At time $t$, the algorithm solves the same convex program as in \Cref{sec:algorithm}, with $t$ in place of the stage index. For 
each supply vertex $j$, we again use the definitions $X_j^{t-1}:=\sum_{d<t}x_j^d$ and $P_j^{t-1}:=\sum_{d<t}x_j^d \cdot f_j^d(x_j^d)$.
On the feasible interval $z\in[0,1-X_j^{t-1}]$, define the non-predicted penalty by
\[
    \frob_j^t(z):=
    \min\left\{1,\frac{1-R+P_j^{t-1}}{1-z}\right\}
\]
for $z < 1$ and $\frob_j^t(1):=1$,
and define the predicted penalty by
\[
    \fadv_j^t(z):=
    \max\left\{0,\frac{\exp(X_j^{t-1}+z-1)-(1-R+P_j^{t-1})}{z}\right\}
\]
for $z > 0$ and $\fadv_j^t(0)$ as the limit when $z \to 0$. 
As in the $k$-stage setting, we use the penalties
\[
    f_j^t(z):=
    \begin{cases}
        \fadv_j^t(z), & \text{if } j\in A_t,\\
        \frob_j^t(z), & \text{otherwise.}
    \end{cases}
\]
We show that the penalties are non-decreasing and take values in $[0,1]$. Hence $(\CP_t)$ is well-defined.

The interpretation of these penalties is the same as in the $k$-stage
algorithm. The exponential curve is the common safety target in the analysis.
The non-predicted penalty does not use it explicitly; instead, whenever the cap
at $1$ is inactive, it makes the local robustness certificate tight, and the
exponential inequalities below show that this update preserves the target. The
predicted penalty is more aggressive and uses the target directly: it is obtained
by solving the equality version of
$1-R+P_j^t \ge \auxfun_\infty(X_j^t)=\exp(X_j^t-1)$.
Thus a predicted vertex is made attractive to the current arrival, but only up to the point where enough protection remains for the unknown future.

The analysis is best viewed as the limit of the $k$-stage proof, not as a fundamentally new
dual-fitting argument. The dual variables are defined from the KKT multipliers
exactly as in \Cref{sec:dual_fitting}, and the same types of load certificates
imply robustness and consistency. What changes is only the scalar invariant
calculus. The two elementary update rules for the finite curve $\auxfun_m$
become the exponential inequalities
$\frac{\exp(z-1)}{1-y}\ge \exp(z+y-1)$
and
$\exp(z-1)+y\ge \exp(z+y-1)$.
These are the direct analogues of the two updates used in the $k$-stage
invariant: an uncapped non-predicted step updates the invariant
multiplicatively, while a capped or predicted step pays for load additively.

The consistency calculation has the same origin. In the $k$-stage proof,
the AM-GM step over $k$ quantities yields $C_k(R)=k(1-R)^{1/k}+R-(k-1)$, 
as in \Cref{Lemma:feasibility_consistency1}. 
In the online setting, this becomes its logarithmic form and gives the prefix certificate $1+R+\ln(1-R)$. 
Thus, the curve in \Cref{thm:online} is the limiting load inequality of the multi-stage analysis. 
The direct proof applies to every finite stopping time $T$, although the algorithm never knows $T$ in advance; see \Cref{app:online}. 

\section{Fractional AdWords with Fractional Predictions}\label{sec:AdWords}

Our framework also extends to the $k$-stage fractional AdWords problem with
fractional predictions.  
We now sketch the main ideas; 
the full proof with all details is
deferred to \Cref{app:AdWords}.

The main difficulty of fractional predictions is that they no longer identify a disjoint set
of predicted supply vertices.  
A single advertiser $j \in S$ can be fractionally
recommended in many stages, and the amount recommended in a stage depends
on bids rather than on a unit-capacity matching structure. 
We handle this
by replacing every advertiser $j \in S$ by stage-indexed \emph{virtual advertisers}.
For every non-final stage $s\in[k-1]$, create a copy $(j,s)$ with budget
$\widehat B_{(j,s)}:=\hat x_j^s B_j$, where
\(
    \hat x_j^s
    :=
    \frac{1}{B_j}
    \sum_{i\in D_s} b_{ij} \hx_{ij}^s
\)
is the fraction of $j$'s budget used by the prediction in stage $s$.
The final copy is a \emph{residual-capacity} copy, not a copy whose budget is
required to equal the spend of the final-stage residual prediction. Define
\(
    \hat\rho_j := 1-\sum_{s=1}^{k-1}\hat x_j^s
\)
and set $\widehat B_{(j,k)}:=\hat\rho_j B_j$.
Thus, $(j,k)$ contains all budget of advertiser $j$ that remains after the
predicted prefix. The completion $\hx^k$, which is chosen only for the
analysis as an optimal allocation in the final residual instance, is lifted to
this residual-capacity copy; it is feasible there but need not fill it.
If $\widehat B_h = 0$ we simply omit the virtual advertiser $h$.
Let $\mathcal H$ denote the set of all virtual advertisers.
Every original edge $(i,j) \in E$ is copied to every virtual advertiser $(j,s)$, $s \in [k]$, with the same bid $b_{i,(j,s)} := b_{ij}$.
For $s<k$, the predicted allocation in stage $s$ is lifted to the copy $(j,s)$;
the final residual allocation $\hx^k$ is lifted to $(j,k)$.
Thus, the lifted prediction exactly fills each non-final predicted copy, while
it may only partially use the residual-capacity copy.  The virtual instance
restores the structural feature used by the matching analysis in the following
sense: every virtual advertiser has a unique designated prediction stage, and
every edge of the lifted prediction in stage $s$ uses only advertisers
designated for stage $s$.

The same construction also covers vertex-weighted fractional bipartite matching
with fractional predictions without introducing any additional budget structure.
Indeed, in the special case $b_{ij}=B_j=w_j$, the virtual budget
$\widehat B_{(j,s)}=\hx_j^s w_j$ is just the value-scaled form of a virtual copy
that contains $\hx_j^s$ units of the unit capacity of supply vertex $j$. When
virtual allocations are aggregated over all copies of $j$, the resulting
allocation uses at most one unit of the original supply capacity, and the value
is unchanged because every copied edge incident to $j$ has value $w_j$.


Any allocation in the virtual instance aggregates over the copies of each advertiser to a feasible allocation in the original AdWords instance with the same value, and conversely any feasible allocation on an original advertiser can be split among its virtual copies subject to their budgets. 
Thus the offline optimum is preserved. 
The prediction benchmark is preserved as well: after the lifted non-final prediction is followed, all non-final copies $(j,s)$, $s<k$, are saturated, leaving for each advertiser $j$ only the residual-capacity copy $(j,k)$, whose budget equals the remaining budget of $j$ in the original residual instance. 
Hence the final residual LP is the same in the virtual and original instances, and therefore $\OPT_{\mathrm{virt}}=\OPT$ and $\advice_{\mathrm{virt}}=\advice$.

We then run our penalty algorithm from \Cref{sec:algorithm} on the virtual instance, using
normalized load
\(
    x_h^s
    :=
    \frac{1}{\widehat B_h}
    \sum_{i\in D_s} b_{ih}x_{ih}^s
\)
for each virtual advertiser $h \in \mathcal H$. 
Formally, in each stage $s \in [k]$, we solve the following convex problem, where $E_s^{\mathrm{virt}}$ denotes the set of all virtual edges in stage $s$:
\begin{alignat}{3}
    (\CP_s^{\mathrm{Ad}}) \quad\max \quad
    &\sum_{h\in\mathcal H}\hB_h\left(x_h^s-\int_0^{x_h^s}f_h^s(t)\,dt\right) \notag \\
    \text{s.t.}\quad
    &\sum_{h:(i,h)\in E_s^{\mathrm{virt}}}x_{ih}^s\le 1 &&\qquad \forall i\in D_s \notag\\
    &x_h^s=\frac{1}{\hB_h}\sum_{i\in D_s:(i,h)\in E_s^{\mathrm{virt}}} b_{ih}x_{ih}^s &&\qquad \forall h\in\mathcal H \notag\\
    &x_h^s\le 1-X_h^{s-1} &&\qquad \forall h\in\mathcal H \notag\\
    &x_{ih}^s\ge 0 &&\qquad \forall (i,h)\in E_s^{\mathrm{virt}} . \notag
\end{alignat}
We use the same penalty functions $a_h^s$ and $r_h^s$ for all $h \in \mathcal H$ as defined in \Cref{sec:penalty-functions}. For each $h \in \mathcal H$ and non-final stage $s \in [k-1]$, we then set 
for all $z \in [0,1-X_h^{s-1}]$
\[
    f_h^s(z):=
    \begin{cases}
        a_h^s(z), & \text{if } h=(j,\tau) \text{ with } \tau=s,\\
        r_h^s(z), & \text{otherwise}
    \end{cases}
\]
and set $f^k_h := 0$ for the final stage.

The analysis is
the same load-based dual fitting as before, with the only change that the
AdWords dual constraint is scaled by bids,
that is, $\alpha_i+ b_{ih} \beta_h / \widehat B_h \ge b_{ih}$.
Equivalently, after substituting the normalized-load equality defining $x_h^s$,
the current-stage capacity constraint for $h$ is
$\sum_i (b_{ih}/\widehat B_h)x_{ih}^s\le 1-X_h^{s-1}$.
The KKT conditions of this reduced stage program have marginal value
$b_{ih}(1-f_h^s(x_h^s))$, exactly matching the scaled dual constraint.
Consequently, the robustness and consistency load inequalities from 
\Cref{lem:conditional-rob} and \Cref{lem:conditional-cons}
apply directly to every virtual advertiser. 
Using our analysis of the penalty functions in \Cref{sec:proofs-main-theorem},
this
gives $\ALG_{\mathrm{virt}}\ge R \cdot \OPT_{\mathrm{virt}}$
and $\ALG_{\mathrm{virt}}\ge C_k(R) \cdot \advice_{\mathrm{virt}}$
for the virtual instance. As we argued before, these bounds immediately translate 
to the original instance. In particular, the solution of the algorithm
on the virtual instance
can be mapped to a solution on the original instance.

A final issue is online implementability: the definition of the virtual instance
requires to know the predictions for each stage a priori. 
We can resolve this by a lazy construction.
Before stage $s$, all unrevealed future copies of advertiser $j$ are
represented by one \emph{residual virtual advertiser} of budget
$(1-\sum_{\tau<s}\hat x_j^\tau)B_j$.
When the stage-$s$ prediction is revealed, this residual advertiser is split
into the newly predicted copy $(j,s)$ and a new residual copy. 
Previously allocated mass is split proportionally. 
This aggregation is exact: if several future copies have
the same history and penalty, replacing them by one residual advertiser
does not change the current convex program. 
This follows from concavity of the objective function of $(\CP_s^{\mathrm{Ad}})$ and Jensen's inequality, and since any aggregate solution
can be expanded proportionally among the copies. Thus the virtual algorithm
is implementable when stages arrive online.
The same lazy splitting idea applies even if the number of stages is unknown to the algorithm; 
everything is analogous and we omit further details here.

\section{Conclusion and Open Questions}

We introduced a penalty-calibration framework for learning-augmented bipartite allocation that yields consistency--robustness tradeoffs in multi-stage, online, and AdWords allocation settings.

Several questions remain open. The most immediate one is to determine the optimal consistency–robustness tradeoff for $k\ge3$. The two-stage upper bound 
\cite{JinM22} applies to our setting, but still leaves a gap to our curve $C_k(R)$ for $k\ge3$. Closing this gap, either by proving a matching upper bound for $C_k(R)$ or by designing algorithms with a stronger tradeoff, is a major open problem.
In particular, showing an improved upper bound over \cite{JinM22}'s
tradeoff upper bound $C_2(R)$ for $k \ge 2$ would be an interesting result.

A separate modeling question concerns the form of predictions. Our model uses an output prediction: a globally feasible allocation that may, for example, be obtained by optimizing against a forecast of the full instance. The prediction is revealed stage by stage but is fixed independently of the algorithm's decisions. In particular, later recommendations are not recomputed against the algorithm's residual capacities. This is the direct multi-stage analogue of the two-stage output benchmark of Jin and Ma~\cite{JinM22}.
It remains open whether comparable guarantees are possible from more \emph{adaptive} prediction models, which may also require different benchmarks and types of guarantees.

Another natural direction is to extend the guarantees to integral online matching. 
Our algorithms are fractional marginal-value algorithms, and preserving their tradeoffs would require more than simply interpreting fractional allocation values as matching probabilities. 
General lossless online rounding of fractional online matching algorithms is impossible~\cite{DevanurJK13}. 
At the same time, lossless rounding can be possible for fractional algorithms with additional structure, as shown in~\cite{DBLP:conf/soda/BuchbinderNW23} for sound two-choice fractional algorithms via additional non-convex constraints. 
Our stage-wise convex-programming algorithms are not two-choice algorithms and their allocations do not obviously satisfy such roundability constraints. 
Whether they admit a lossless randomized implementation, or whether similar tradeoffs can be achieved by a different integral algorithm, remains open.



\section*{Acknowledgments}

This work was supported by the Deutsche Forschungsgemeinschaft (DFG, German Research Foundation) through project no. 547924951 and under Germany's Excellence Strategy EXC-3036 project no. 533607631 and EXC-2046/2.
Part of this work was done while AL and NM were participants in the Simons Institute program on Algorithmic
Foundations for Emerging Computing Technologies.

\printbibliography

\newpage 
\appendix
\crefalias{section}{appendix}
\crefalias{subsection}{appendix}
\crefalias{subsubsection}{appendix}
\section*{Appendices}

\section{Omitted Proofs from \Cref{sec:algorithm} (Algorithm)}
\label{app:algorithm}

\lemmaKKT*

\begin{proof}
    Let $s \in [k]$.
    There exist non-negative numbers $(\lambda^s_i)_{ i \in D_s}$, $(\theta^s_j)_{ j \in S}$, and $(\gamma^s_{ij})_{ (i,j) \in E_s}$
	that satisfy the following KKT conditions~\cite{BV2014}.
	For each $(i,j) \in E_s$, stationarity gives
	 \[
		 w_j(1-f^s_j(x^s_j)) = \lambda_i^s + \theta_j^s - \gamma_{ij}^s \ , 
	\]
	and complementary slackness gives
		\begin{align*}
			\lambda_i^s(1-x_i^s) &=0 \quad \forall i \in D_s \ , \\
			\theta_j^s(1- X_j^{s-1} - x_j^s) &= 0 \quad \forall j \in S \ , \text{ and} \\
			\gamma_{ij}^s x_{ij}^s &= 0 \quad \forall (i,j) \in E_s \ .
		\end{align*}

	We now show that $(\lambda_i^s)_{i \in D_s}$ satisfies the properties of the lemma.
	\begin{enumerate}[label=(\alph*),nosep]
		\item Let $(i,j) \in E_s$ be an edge with $x_{ij}^s > 0$. By complementary slackness, $\gamma_{ij}^s = 0$, so stationarity gives $w_j (1-f^s_j(x_j^s)) = \lambda_i^s + \theta_j^s \geq \lambda_i^s$.
		
		\item Consider $(i,j),(m,j) \in E_s$ with $x_{mj}^s > 0$. By complementary slackness for the used edge $(m,j)$, we have $\gamma_{mj}^s = 0$. Hence stationarity for $(m,j)$ gives
		\[
		    w_j(1-f_j^s(x_j^s)) = \lambda_m^s + \theta_j^s .
		\]
		For edge $(i,j)$, stationarity gives
		\[
		    w_j(1-f_j^s(x_j^s)) = \lambda_i^s + \theta_j^s - \gamma_{ij}^s \leq \lambda_i^s + \theta_j^s,
		\]
		because $\gamma_{ij}^s\geq 0$. Therefore, $\lambda_i^s \geq \lambda_m^s$.
		
		\item If $x_i^s < 1$ for some $i \in D_s$, then $\lambda_i^s = 0$ by complementary slackness.
		
		\item Consider $(i,j) \in E_s$ 
		with $X_j^s < 1$. Since $X_j^{s-1} + x_j^s = X_j^s < 1$, the capacity constraint $x_j^s \leq 1 - X_j^{s-1}$ is strictly inactive. By complementary slackness, 	we have $\theta^s_j = 0$. Thus $w_j (1-f^s_j(x_j^s)) = \lambda_i^s-\gamma_{ij}^s \leq \lambda_i^s$.
	\end{enumerate} 
\end{proof}

\section{Omitted Proofs from \Cref{sec:dual_fitting} (Dual Fitting Analysis)}

\subsection{Omitted Proofs from \Cref{sec:dual-objective} (Dual Objective Value)}
\label{app:dual_objective}

\lemmaDualObjective*

\begin{proof}
For every $j\in S$, the definition of $\beta_j$ gives
\[
    \beta_j
    +
    \sum_{s=1}^{k}\sum_{i\in D_s}\alpha_i x_{ij}^s
    =
    w_jX_j^k .
\]
Therefore,
\begin{align*}
    \sum_{j\in S}\beta_j+\sum_{i\in D}\alpha_i
    &=
    \sum_{j\in S}\beta_j + \sum_{s=1}^{k}\sum_{i\in D_s}\alpha_i x_i^s \\
    &=
    \sum_{j\in S}
    \biggl(
        \beta_j + \sum_{s=1}^{k}\sum_{i\in D_s}\alpha_i x_{ij}^s
    \biggr) 
    =
    \sum_{j\in S}w_jX_j^k
    =
    \ALG .
\end{align*}
The first equality uses \Cref{lem:kkt}(c):
if $x_i^s<1$, then $\alpha_i=\lambda_i^s=0$, while if $x_i^s=1$ then
$\alpha_i=\alpha_i x_i^s$.
\end{proof}

\subsection{Omitted Proofs from \Cref{sec:dual-feasibility} (Dual Feasibility)}
\label{app:dual_feasibility}

\lemmaDualConstraintLowerBound*

\begin{proof}
First, for every $j\in S$,
\begin{align*}
    \beta_j
    &=
    \sum_{d=1}^{k}
    \biggl(
        w_jx_j^d-\sum_{i\in D_d}\lambda_i^d x_{ij}^d
    \biggr) \\
    &\ge
    w_j\sum_{d=1}^{k}
    x_j^d f_j^d(x_j^d)
    =
    w_jP_j^k
    =
    w_jP_j^{k-1},
\end{align*}
where the inequality uses \Cref{lem:kkt}(a)
on every edge with positive allocation, and the last equality uses
$f_j^k = 0$.

Now consider an edge $(i,j)\in E_s$ with $s\in[k]$.

If $X_j^s<1$, then \Cref{lem:kkt}(d) gives
\[
    \alpha_i=\lambda_i^s
    \ge
    w_j(1-f_j^s(x_j^s)).
\]
Together with $\beta_j\ge w_jP_j^{k-1}$, this implies
\[
    \alpha_i+\beta_j
    \ge
    w_j\left(P_j^{k-1}+1-f_j^s(x_j^s)\right).
\]

Next suppose that $X_j^s=1$. Then $q_j\le s$.

If $q_j<s$, then $j$ was already saturated before stage $s$, so
$x_j^d=0$ for every $d>q_j$. Hence $P_j^{k-1}=P_j^{q_j}$, and
\[
    \alpha_i+\beta_j
    \ge
    \beta_j
    \ge
    w_jP_j^{k-1}
    =
    w_jP_j^{q_j}.
\]

It remains to consider the case $q_j=s$. Let $\theta_j^s$ be the KKT
multiplier for the supply constraint of $j$ in stage $s$ (see proof of \Cref{lem:kkt}). 
By
stationarity, for every edge $(r,j)\in E_s$,
\[
    \lambda_r^s
    \ge
    w_j(1-f_j^s(x_j^s))-\theta_j^s .
\]
Moreover, for every used edge $(m,j)$ with $x_{mj}^s>0$,
complementary slackness gives
\[
    \lambda_m^s
    =
    w_j(1-f_j^s(x_j^s))-\theta_j^s .
\]
Therefore, the stage-$s$ contribution to $\beta_j$ is
\[
    w_jx_j^s-\sum_{m\in D_s}\lambda_m^s x_{mj}^s
    =
    w_jx_j^s f_j^s(x_j^s)+\theta_j^s x_j^s .
\]
All earlier stages contribute at least $w_jP_j^{s-1}$, and all later
stage contributions are nonnegative. Thus
\begin{align*}
    \alpha_i+\beta_j
    &\ge
    \lambda_i^s
    +
    w_jP_j^{s-1}
    +
    w_jx_j^s f_j^s(x_j^s)
    +
    \theta_j^s x_j^s \\
    &\ge
    w_j(1-f_j^s(x_j^s))-\theta_j^s
    +
    w_jP_j^{s-1}
    +
    w_jx_j^s f_j^s(x_j^s)
    +
    \theta_j^s x_j^s \\
    &=
    w_j\left(
        1-f_j^s(x_j^s)
        +
        x_j^s f_j^s(x_j^s)
        +
        P_j^{s-1}
    \right)
    -
    \theta_j^s(1-x_j^s).
\end{align*}
Since $q_j=s$, some edge $(m,j)$ with $x_{mj}^s>0$ exists. Hence
\[
    0\le \lambda_m^s
    =
    w_j(1-f_j^s(x_j^s))-\theta_j^s,
\]
so
\[
    \theta_j^s\le w_j(1-f_j^s(x_j^s)).
\]
Substituting this bound yields
\begin{align*}
    \alpha_i+\beta_j
    &\ge
    w_j\left(
        1-f_j^s(x_j^s)
        +
        x_j^s f_j^s(x_j^s)
        +
        P_j^{s-1}
    \right)
    -
    w_j(1-f_j^s(x_j^s))(1-x_j^s) \\
    &=
    w_j\left(P_j^{s-1}+x_j^s\right).
\end{align*}

Finally, let $(i,j)\in E_k$. If $X_j^k<1$, the first case gives
\[
    \alpha_i+\beta_j
    \ge
    w_j(P_j^{k-1}+1)
    \ge
    w_j(P_j^{k-1}+1-X_j^{k-1}).
\]
If $X_j^k=1$ and $q_j<k$, then $X_j^{k-1}=1$, and the second case gives
\[
    \alpha_i+\beta_j
    \ge
    w_jP_j^{q_j}
    =
    w_jP_j^{k-1}
    =
    w_j(P_j^{k-1}+1-X_j^{k-1}).
\]
If $X_j^k=1$ and $q_j=k$, then $x_j^k=1-X_j^{k-1}$, and the third case gives
\[
    \alpha_i+\beta_j
    \ge
    w_j(P_j^{k-1}+x_j^k)
    =
    w_j(P_j^{k-1}+1-X_j^{k-1}).
\]
This proves the final-stage bound.
\end{proof}

\lemmaConditionalRobustness*

\begin{proof}
    We first show that if $\alpha_i + \beta_j \geq R w_j$ for all $(i,j) \in E$, then the algorithm is $R$-robust.
    If $R=0$, this is immediate. Hence assume $R>0$.
    Then scaling our dual solution down by a factor of $R$ gives a feasible solution to $(\DP)$.
    By \Cref{lem:dual_objective} and weak duality, we conclude
    \[
        \frac{1}{R} \alg 
        = 
        \sum_{i \in D} \frac{\alpha_i}{R} + \sum_{j \in S} \frac{\beta_j}{R}
        \geq
        \DP(G) \geq \LP(G) = \opt \ .
    \]
    Thus $\alg \geq R \cdot \opt$, so the algorithm is $R$-robust.

	It remains to prove the claimed approximate feasibility.
	Let $s \in [k-1]$ and $(i,j) \in E_s$. One of the following cases applies:
	\begin{itemize}
		\item If $X_j^s < 1$, then $\alpha_i + \beta_j \geq w_j (P_j^{k-1} + 1 - f_j^s(x_j^s)) \geq R w_j$ by \Cref{lem:intermed_feasibility} and \eqref{eq:k-rob-s}.
		\item If $q_j < s$, then note that $X_j^{k-1} = 1$ implies $P_{j}^{k-1} + 1 - X_j^{k-1} = P_j^{k-1}$. Thus,
		by \Cref{lem:intermed_feasibility} and \eqref{eq:k-rob-k},
        we have $\alpha_i + \beta_j \geq w_j P_j^{q_j} = w_j P_j^{k-1} = w_j \left( 1 - X_j^{k-1} + P_j^{k-1} \right) \geq w_j R$.
        \item Similarly, if $q_j = s$ then
        $\alpha_i + \beta_j \geq w_j (P_j^{s-1} + x_j^s) \geq w_j P_j^s = w_j P_j^{q_j} \geq R w_j$ by \Cref{lem:intermed_feasibility} and \eqref{eq:k-rob-k}, because $f_j^s(x_j^s) \leq 1$. 
	\end{itemize}

	Finally, for any \((i,j)\in E_k\), \Cref{lem:intermed_feasibility} gives
\[
    \alpha_i+\beta_j
    \ge
    w_j\left(P_j^{k-1}+1-X_j^{k-1}\right)
    \ge
    Rw_j,
\]
where the last inequality uses \eqref{eq:k-rob-k}.	
\end{proof}

\lemmaConditionalConsistency*

\begin{proof}
	Let $\widehat M= \widehat M_1\cup\cdots\cup \widehat M_k$ be the predicted matching (as a set of edges). We show that
	\begin{equation}
	    \alpha_i+\beta_j\ge C_k(R)w_j
	    \tag{$*$}\label{eq:predicted-edge-dual}
	\end{equation}
	for all $(i,j)\in \widehat M$.
	Since $\widehat M$ is a matching and all dual variables are nonnegative,
	\begin{align*}
	    \ALG
	    =
	    \sum_{i\in D}\alpha_i+\sum_{j\in S}\beta_j 
	    \ge
	    \sum_{(i,j)\in \widehat M}(\alpha_i+\beta_j) 
	    \overset{\eqref{eq:predicted-edge-dual}}{\ge}
	    C_k(R)\sum_{(i,j)\in \widehat M}w_j 
	    =
	    C_k(R)\cdot \advice \ ,
	\end{align*}
	as desired.
	It remains to prove \eqref{eq:predicted-edge-dual} for every predicted edge.

	To this end, consider an edge $(i,j) \in \widehat M_s$ for some stage $s \in [k-1]$. Since $\widehat M_s$ is the predicted matching for stage $s$, we have $j\in A_s$. One of the following cases applies:
	\begin{itemize}
		\item If $X_j^s < 1$, then
        \[
            \alpha_i + \beta_j \geq w_j(P_j^{k-1} + 1 - f_j^s(x_j^s)) \geq w_j \cdot C_k(R)
        \]
        using \Cref{lem:intermed_feasibility} and \eqref{eq:k-cons-s}.
		\item If $q_j = s$, then $x_j^s = 1 - X_j^{s - 1}$, so
		\[
		\alpha_i + \beta_j \geq w_j (P_j^{s-1} + x_j^s) = w_j(P_j^{s-1} + 1 - X_j^{s-1}) \geq w_j \cdot C_k(R) 
		\]
		using \Cref{lem:intermed_feasibility} and \eqref{eq:k-cons-prefix} with $p=s-1$.
		\item If $q_j < s$, then 
		\[
		\alpha_i + \beta_j \geq w_j P_j^{q_j} = w_j (P_j^{q_j} + 1 - X_j^{q_j}) \geq w_j \cdot C_k(R) 
		\]
		using \Cref{lem:intermed_feasibility} and \eqref{eq:k-cons-prefix} with $p=q_j$.
	\end{itemize}

	Finally, consider a predicted final-stage edge $(i,j)\in \widehat M_k$. Then
	$j\in A_k$. By \Cref{lem:intermed_feasibility},
	\[
	    \alpha_i+\beta_j
	    \ge
	    w_j\left(P_j^{k-1}+1-X_j^{k-1}\right).
	\]
	Using \eqref{eq:k-cons-prefix} with $s=k$ and $p=k-1$, we get
	\[
	    P_j^{k-1}+1-X_j^{k-1}\ge C_k(R).
	\]
	Therefore
	\[
	    \alpha_i+\beta_j\ge C_k(R)w_j.
	\]
\end{proof}

\subsection{Omitted Proofs from 
\Cref{sec:penalty-functions} (Properties of Penalty Functions)}
\label{app:properties-penalties}

\subsubsection{Properties of the Safety Curve}\label{app:reserve-props}

\lemmaLambdaProps*

\begin{proof}
For part (a),
the upper bound $\auxfun_s(z) \leq 1$ follows directly from the definition of $\auxfun_s(z)$. For the lower bound,
we have
\[
\auxfun_s(z) = \left( \frac{s-1+z}{s} \right)^s = \left( \frac{1}{s} \left(z + \sum_{i=1}^{s-1} 1 \right) \right)^s \geq \left( \left(z \cdot \prod_{i=1}^{s-1} 1 \right)^{1/s} \right)^s = z
\]
using \Cref{am-gm}.

For part (b), note that
\begin{align*}
	(1-y) \cdot \auxfun_{s-1}(z+y)
	&=
  (1-y) \cdot \prod_{i=1}^{s-1} \frac{(s-1) - 1 + (z+y)}{s-1} \\
  &\leq \left( \frac{1}{s} \left(1-y + \sum_{i=1}^{s-1} \frac{s-2+z+y}{s-1} \right) \right)^s \\
  &= \left( \frac{s-1+z}{s} \right)^s \\
  &= \auxfun_s(z)
\end{align*}
using \Cref{am-gm}.

For part (c), define
\[
L(y) := \auxfun_s(z)+y-\auxfun_{s-1}(z+y).
\]
Then, it holds for the derivative that
\[
L'(y) = 1-\left(\frac{s-2+z+y}{s-1}\right)^{s-2}\ge 0,
\]
because $z+y \leq 1$. Thus $L$ is non-decreasing in $y$ (and $L'(1-z)=0$). Hence
\[
L(y)\ge L(0)=\auxfun_s(z)-\auxfun_{s-1}(z)\ge 0,
\]
using the inequality in (b) with $y=0$.
\end{proof}

\subsubsection{Validity of the Penalty Functions}
\label{app:validity-penalties}

In this section we prove that the penalty functions used by the algorithm satisfy the assumptions of \Cref{lem:conditional-convexity}. Namely, we show that for every stage $s\in[k]$ and every supply vertex $j\in S$, the function $f_j^s$ is continuous, non-decreasing, and takes values in $[0,1]$ on its domain $[0,1-X_j^{s-1}]$.

\LemmaPenaltiesValid*

The proof proceeds in three steps. First, we establish an invariant for the accumulated penalty load, namely $1-R+P_j^s\ge g_{k-s}(X_j^s)$. This invariant implies monotonicity of the predicted penalty branch $a_j^s$. The non-predicted branch $r_j^s$ is monotone directly from its definition. Finally, the same invariant and the basic properties of $g_m$ imply that both branches take values in $[0,1]$.


For the following proof, we set $P_j^0 = 0$ and $X_j^0 = 0$ for all $j \in S$.

\begin{lemma}\label{lem:g-invariant}
	For all $s \in \{0,\ldots,k-1\}$ and $j \in S$, it holds that $1 - R + P_j^s \geq \auxfun_{k-s}(X^s_j)$. 
\end{lemma}

\begin{proof}
	We show the lemma via induction over $s$. If $s=0$, we have that $1-R \geq 1 - R_k = (1-1/k)^k = \auxfun_{k-0}(X_j^0)$.

	Now let $s \geq 1$ and suppose that $1 - R + P_j^{s-1} \geq \auxfun_{k-s+1}(X_j^{s-1})$ holds.
	We distinguish between two main cases:

	\medskip
	\noindent
	\textbf{Case 1:} If $j \notin A_s$, then $f_j^s = r_j^s$. If $r_j^s(x_j^s) = 1$, then
	\[
		1-R+P_j^s
		=
		1 - R + P_j^{s-1} + x_j^s r_j^s(x_j^s) 
		=
		1 - R + P_j^{s-1} + x_j^s 
		\geq
		\auxfun_{k-s+1}(X_j^{s-1}) + x_j^s
		\geq 
			\auxfun_{k-s}(X_j^{s})  \ ,
	\] 
	where the first inequality uses the induction hypothesis and the second uses \Cref{lem:g-props}(c) with
	$m=k-s+1$, $z=X_j^{s-1}$, and $y=x_j^s$. Here $m\in\{2,\ldots,k\}$ because $s\in[k-1]$, and feasibility gives $y\in[0,1-z]$.
	If $r_j^s(x_j^s) < 1$, then
	\begin{align*}
		1-R+P_j^s
		&=
		1 - R + P_j^{s-1} + x_j^s r_j^s(x_j^s) \\
		&=
		1 - R + P_j^{s-1} + x_j^s \frac{1 - R + P_j^{s-1}}{1-x_j^s} \\
		&=
		\frac{1 - R + P_j^{s-1}}{1-x_j^s} \\
		&\geq
		\frac{\auxfun_{k-s+1}(X_j^{s-1})}{1-x_j^s}
		\geq \auxfun_{k-s}(X_j^{s}) \ ,
\end{align*}
	where the first inequality uses the induction hypothesis. For the second, apply \Cref{lem:g-props}(b) with
	$m=k-s+1$, $z=X_j^{s-1}$, and $y=x_j^s$. Feasibility gives $y\in[0,1-z]$, while $r_j^s(x_j^s)<1$ gives $y<1$, so dividing the resulting inequality by $1-y$ is valid.

	\medskip
	\noindent
	\textbf{Case 2:} If $j \in A_s$, so $f_j^s = a_j^s$.
	If $x_j^s = 0$, then
	\[
	1 - R + P_j^{s-1} + x_j^s a_j^s(x_j^s) 
	= 
	1 - R + P_j^{s-1}
	\geq
	\auxfun_{k-s+1}(X_j^{s-1}) 
	\geq
	\auxfun_{k-s}(X_j^{s-1}) \ ,
	\]
	where the first inequality uses the induction hypothesis and the second inequality uses \Cref{lem:g-props}(c) with $y=0$.
	If $x_j^s > 0$ and $a_j^s(x_j^s) = 0$, then by definition of $a_j^s$, it must hold that $\auxfun_{k-s}(X_j^{s-1} + x_j^s) - (1 - R + P_j^{s-1}) \leq 0$, so
	\[
	1 - R + P_j^{s-1} + x_j^s a_j^s(x_j^s) 
	=
	1 - R + P_j^{s-1} 
	\geq
	\auxfun_{k-s}(X_j^{s}) \ .
	\]
	Finally, if $x_j^s > 0$ and $a_j^s(x_j^s) > 0$, then by the definition of $a_j^s$, we have
	\[
	1 - R + P_j^{s-1} + x_j^s a_j^s(x_j^s)
	=
	1 - R + P_j^{s-1} + \auxfun_{k-s}(X_j^{s-1} + x_j^s) - (1 - R + P_j^{s-1})
	= \auxfun_{k-s}(X_j^{s}) \ .
	\]

	This completes the proof of the invariant.
\end{proof}

\begin{lemma}\label{lem:cons-pen-monotone}
    The function $a_j^s$ is non-decreasing on $[0,1-X_j^{s-1}]$ for all $s \in [k-1]$ and $j \in A_s$.
\end{lemma}

\begin{proof}
	Fix $s \in [k-1]$ and $j \in A_s$. 
	By \Cref{lem:g-invariant} and \Cref{lem:g-props}(b) we have
	\begin{equation}
	    1 - R + P_j^{s-1} 
		\geq \auxfun_{k-s+1}(X_j^{s-1})
		\geq \auxfun_{k-s}(X_j^{s-1}) . \label{eq:monotonicity-1}
	\end{equation}
    For $z>0$, let
	\[
	    h(z) := \frac{\auxfun_{k-s}(z + X_j^{s-1}) + R - 1 - P_j^{s-1}}{z}.
	\]
	Computing the derivative gives
	\[
	    h'(z)=\frac{1 - R + P_j^{s-1} - \left(\auxfun_{k-s}(X_j^{s-1}+z)-z\auxfun'_{k-s}(X_j^{s-1}+z)\right)}{z^2}.
	\]
	We now focus on
    \[
        \phi(z) := \auxfun_{k-s}(X^{s-1}_j + z) - z \auxfun_{k-s}'(X^{s-1}_j + z).
    \]
    Since $\auxfun_{k-s}$ is convex on $[0,1]$, we have $\auxfun''_{k-s}(y) \geq 0$ for all $y \in [0,1]$. Thus,
	\[
	    \phi'(z) = \auxfun_{k-s}'(X_j^{s-1}  + z) - \auxfun_{k-s}'(X_j^{s-1}  + z) - z \auxfun_{k-s}''(X_j^{s-1}  + z) \leq 0 .
	\]
	This means that $\phi$ is non-increasing, and thus,
	\[
	\phi(z) \leq \phi(0) = \auxfun_{k-s}(X_j^{s-1}) \leq 1 - R + P_j^{s-1}
	\]
	using \eqref{eq:monotonicity-1}. This implies that $h'(z) \geq 0$, so $h$ is non-decreasing on $(0,1-X_j^{s-1}]$. 
	Since $a_j^s(z)=\max\{0,h(z)\}$ for $z>0$, the function $a_j^s$ is non-decreasing on $(0,1-X_j^{s-1}]$. Finally, $a_j^s(0)=0\leq a_j^s(z)$ for every $z>0$, so $a_j^s$ is non-decreasing on the entire feasible interval.
\end{proof}

\begin{lemma}\label{lem:rob-pen-monotone}
	The function $r_j^s$ is non-decreasing on $[0,1-X_j^{s-1}]$ for all $s \in [k-1]$ and $j \in S$.
\end{lemma}

\begin{proof}
    By \Cref{lem:g-invariant}, it holds that
    $1-R+P_j^{s-1}\geq \auxfun_{k-s+1}(X_j^{s-1})\geq 0$.
    Therefore the function
    \[
        z\mapsto \frac{1-R+P_j^{s-1}}{1-z}
    \]
    is non-decreasing on $[0,1)$. Taking the minimum with $1$ preserves monotonicity, and the definition $r_j^s(1)=1$ is consistent with monotonicity at the endpoint. Hence $r_j^s$ is non-decreasing on the feasible interval.
\end{proof}

\begin{lemma}\label{lem:penalty-in-0-1}
	For all $s \in [k-1]$ and $j \in S$, it holds that $f_j^s(z) \in [0,1]$ for all $z \in [0, 1 - X_j^{s-1}]$.
\end{lemma}

\begin{proof}
    Fix a stage $s \in [k-1]$, a vertex $j \in S$, and a value $z \in [0, 1 - X_j^{s-1}]$. We show that $0 \leq f_j^s(z) \leq 1$ by analyzing the two cases for the penalty function.
    
    \medskip
    \noindent
    \textbf{Case 1: $j \notin A_s$.} In this case, $f_j^s(z) = r_j^s(z)$. The upper bound $r_j^s(z)\leq 1$ follows directly from the definition. For the lower bound, if $z=1$, then $r_j^s(1)=1$. If $z<1$, then $1-z>0$, and \Cref{lem:g-invariant} gives
    \[
        1-R+P_j^{s-1}\geq \auxfun_{k-s+1}(X_j^{s-1})\geq 0.
    \]
    Hence $(1-R+P_j^{s-1})/(1-z)\geq 0$, and therefore $r_j^s(z)\geq 0$.
        
    \medskip
    \noindent
    \textbf{Case 2: $j \in A_s$.} Here, $f_j^s(z) = a_j^s(z)$. The lower bound $a_j^s(z) \geq 0$ is immediate from the definition. To establish the upper bound $a_j^s(z) \leq 1$, it is enough to consider $z > 0$, since $a_j^s(0)=0$. By \Cref{lem:g-invariant},
    \[
        1 - R + P_j^{s-1} \geq \auxfun_{k-s+1}(X_j^{s-1}).
    \]
    Moreover, by \Cref{lem:g-props}(c),
    \[
        \auxfun_{k-s+1}(X_j^{s-1}) + z \geq \auxfun_{k-s}(X_j^{s-1}+z).
    \]
    Combining these two inequalities gives
    \[
        1 - R + P_j^{s-1} + z \geq \auxfun_{k-s}(X_j^{s-1}+z).
    \]
    Hence
    \[
        \frac{\auxfun_{k-s}(X_j^{s-1}+z)-(1-R+P_j^{s-1})}{z}\leq 1.
    \]
    Taking the maximum with $0$ preserves the upper bound, so $a_j^s(z)\leq 1$.
\end{proof}

\begin{proof}[Proof of \Cref{lem:penalties-valid}]
For $s=k$, we have $f_j^k = 0$, so the claim is immediate.
The non-final stages $s\in[k-1]$ are handled by \Cref{lem:g-invariant,lem:cons-pen-monotone,lem:rob-pen-monotone,lem:penalty-in-0-1}.
It remains only to check continuity at the defined endpoints. The function $r_j^s$ is continuous on $[0,1)$, and if its feasible domain reaches $1$, then
$1-R+P_j^{s-1}>0$ implies that $r_j^s(z)=1$ for all $z$ sufficiently close to $1$. Hence its extension $r_j^s(1)=1$ is continuous.
For $a_j^s$, set $m=k-s$ and $X=X_j^{s-1}$. If $X=1$, its domain is the singleton $\{0\}$. If $X<1$, then \Cref{lem:g-invariant} and the strict form of \Cref{lem:g-props}(b) at $y=0$ give
\[
    1-R+P_j^{s-1}\ge g_{m+1}(X)>g_m(X).
\]
By continuity of $g_m$, the numerator defining $a_j^s(z)$ is therefore negative for all sufficiently small $z>0$, so $a_j^s(z)=0$ near zero. Thus $a_j^s(0)=0$ is its continuous extension. Continuity elsewhere follows directly from the formula. Consequently every $f_j^s$ is continuous on its feasible interval.
\end{proof}

\subsection{Omitted Proofs from \cref{sec:proofs-main-theorem} (Proof of Main Theorem)}
\label{app:proofs-conditions}

\lemmaUncappedNonpredicted*

\begin{proof}
For $h=0$, the lemma is trivial. Now assume $h\geq 1$.
If $r_j^d(x_j^d)<1$, then $x_j^d<1$ and the cap in the definition of $r_j^d$ is inactive. Hence
\[
    r_j^d(x_j^d)=\frac{1-R+P_j^{d-1}}{1-x_j^d}.
\]
Therefore
\begin{align*}
    1-R+P_j^d
    &=
    1-R+P_j^{d-1}+x_j^d r_j^d(x_j^d) \\
    &=
    1-R+P_j^{d-1}+x_j^d\frac{1-R+P_j^{d-1}}{1-x_j^d} \\
    &=\frac{1-R+P_j^{d-1}}{1-x_j^d}.
\end{align*}
Iterating this recurrence from $d=1$ to $h$ proves the first statement.

For the second statement, if $x_j^t<1$, then $r_j^t(x_j^t)=1$ implies
\[
    \frac{1-R+P_j^{t-1}}{1-x_j^t}\geq 1,
\]
and hence
\[
    1-R+P_j^t=1-R+P_j^{t-1}+x_j^t\geq 1.
\]
If $x_j^t=1$, then $1-R+P_j^t=1-R+P_j^{t-1}+1\geq1$ as well.
We now induct over $d=t+1,\ldots,h$, maintaining $1-R+P_j^{d-1}\geq1$.
If $x_j^d<1$, then
\[
    \frac{1-R+P_j^{d-1}}{1-x_j^d}\geq 1,
\]
so $r_j^d(x_j^d)=1$. If $x_j^d=1$, then $r_j^d(x_j^d)=1$ by definition. In either case,
$P_j^d=P_j^{d-1}+x_j^d\ge P_j^{d-1}$, so the induction invariant also holds for the next index.
Thus $r_j^d(x_j^d)=1$ for all $d\in\{t,\ldots,h\}$. During these stages, $P_j^d-X_j^d$ does not change, since its increment is $x_j^d \cdot r_j^d(x_j^d)-x_j^d=0$. This proves the second statement.
\end{proof}

\lemmaFeasibilityConsistencyTwo*

\begin{proof}
Fix $s\in[k-1]$ and $j\in A_s$.
Since $f_j^s = a_j^s$ and $x_j^d f_j^d(x_j^d) \geq 0$ for all $d \in [k-1]$,
\begin{align*}
    P_j^{k-1}+1-f_j^s(x_j^s)
    \geq P_j^s+1-a_j^s(x_j^s) 
    = 1+P_j^{s-1}+(x_j^s-1)a_j^s(x_j^s).
\end{align*}
Thus it suffices to prove
\begin{equation}\label{eq:stage-consistency-suffices}
    1+P_j^{s-1}+(x_j^s-1)a_j^s(x_j^s) \geq C_k(R).
\end{equation}
If $a_j^s(x_j^s)=0$, then the left-hand side of \eqref{eq:stage-consistency-suffices} is
$1+P_j^{s-1}\geq 1\geq C_k(R)$.
Hence assume that $a_j^s(x_j^s)>0$. In particular, $x_j^s>0$.

Note that $f_j^d = r_j^d$ for all $d \in [s-1]$.
We first show that $r_j^d(x_j^d) < 1$ for all $d \in [s-1]$. Suppose for contradiction that
$r_j^t(x_j^t)=1$ for some $t \leq s-1$. 
If $x_j^t<1$, then the definition of $r_j^t$ gives
$1-R+P_j^{t-1}\geq 1-x_j^t$, and hence $P_j^{t-1}+x_j^t\geq R$.
If $x_j^t=1$, the same inequality $P_j^{t-1}+x_j^t\geq R$ is immediate.
Since $r_j^t(x_j^t)=1$, we get $P_j^t\geq R$, and therefore $P_j^{s-1}\geq R$.
Here the last implication uses that $P_j^d$ is non-decreasing in $d$, since
$P_j^d-P_j^{d-1}=x_j^d f_j^d(x_j^d)\geq0$ for every $d\in[k-1]$.
But then the numerator in the definition of $a_j^s(x_j^s)$ satisfies
\[
    \auxfun_{k-s}(X_j^s)-(1-R+P_j^{s-1})
    \leq
    \auxfun_{k-s}(X_j^s)-1
    \leq 0,
\]
contradicting $a_j^s(x_j^s)>0$.

Since we have shown that $r_j^d(x_j^d)<1$ for every $d\in[s-1]$, \Cref{claim:uncapped-nonpredicted-product} with $h=s-1$ gives
\[
    1-R+P_j^{s-1}
    =
    \frac{1-R}{\prod_{\ell=1}^{s-1}(1-x_j^\ell)} \ .
\]
Subtracting $1-R$ from both sides yields
\begin{equation}\label{eq:stage-consistency-product-S}
    P_j^{s-1}
    =
    \frac{1-R}{\prod_{\ell=1}^{s-1}(1-x_j^\ell)}-(1-R) \ ,
\end{equation}
and, since $a_j^s(x_j^s)>0$, we have
\begin{equation}\label{eq:stage-consistency-product-a}
    x_j^s \cdot a_j^s(x_j^s)
    =
    \auxfun_{k-s}(X_j^s)
    -
    \frac{1-R}{\prod_{\ell=1}^{s-1}(1-x_j^\ell)} \ .
\end{equation}
For clarity, set $Q:=\frac{1-R}{\prod_{\ell=1}^{s-1}(1-x_j^\ell)}$.
Then \eqref{eq:stage-consistency-product-S} and
\eqref{eq:stage-consistency-product-a} say that
$P_j^{s-1}=Q-(1-R)$ and $x_j^s a_j^s(x_j^s)=\auxfun_{k-s}(X_j^s)-Q$.
Using also $C_k(R)=k(1-R)^{1/k}+R-(k-1)$, we obtain
\begin{align}
&x_j^s\left(1+P_j^{s-1}+(x_j^s-1)a_j^s(x_j^s)-C_k(R)\right) \notag\\
&=x_j^s(1-C_k(R))+x_j^s\bigl(Q-(1-R)\bigr)
  +(x_j^s-1)(\auxfun_{k-s}(X_j^s)-Q) \notag\\
&=Q-\auxfun_{k-s}(X_j^s)+x_j^s\bigl(\auxfun_{k-s}(X_j^s)+R-C_k(R)\bigr) \notag\\
&= \frac{1-R}{\prod_{\ell=1}^{s-1}(1-x_j^\ell)}
   -\auxfun_{k-s}(X_j^s)
   +x_j^s\left(\auxfun_{k-s}(X_j^s)+k-k(1-R)^{1/k}-1\right). \label{eq:stage-consistency-main-expression}
\end{align}

Next we lower-bound the first term in \eqref{eq:stage-consistency-main-expression}.
Since $x_j^s>0$, we have $\auxfun_{k-s}(X_j^s)>0$.
We apply \Cref{am-gm} to the $k$ non-negative numbers
\[
    1-x_j^1,\ldots,1-x_j^{s-1},
    \underbrace{\frac{k-s-1+X_j^s}{k-s},\ldots,
    \frac{k-s-1+X_j^s}{k-s}}_{k-s\text{ copies}},
    \frac{1-R}{\left(\prod_{\ell=1}^{s-1}(1-x_j^\ell)\right)\auxfun_{k-s}(X_j^s)},
\]
whose product is $1-R$. Indeed, the product of these $k$ numbers is
\[
\biggl(\prod_{\ell=1}^{s-1}(1-x_j^\ell)\biggr)
\left(\frac{k-s-1+X_j^s}{k-s}\right)^{k-s}
\frac{1-R}{\left(\prod_{\ell=1}^{s-1}(1-x_j^\ell)\right)\auxfun_{k-s}(X_j^s)}
=1-R,
\]
because
\[
\auxfun_{k-s}(X_j^s)=\left(\frac{k-s-1+X_j^s}{k-s}\right)^{k-s}.
\]
By \Cref{am-gm}, we obtain
\begin{align*}
    (1-R)^{1/k}
    &\leq
    \frac{1}{k}\left(
        \sum_{\ell=1}^{s-1}(1-x_j^\ell)
        +(k-s)\frac{k-s-1+X_j^s}{k-s}
        +\frac{1-R}{\left(\prod_{\ell=1}^{s-1}(1-x_j^\ell)\right)\auxfun_{k-s}(X_j^s)}
    \right) \\
    &=
    \frac{1}{k}\left(
        k-2+x_j^s
        +\frac{1-R}{\left(\prod_{\ell=1}^{s-1}(1-x_j^\ell)\right)\auxfun_{k-s}(X_j^s)}
    \right).
\end{align*}
Rearranging yields
\begin{equation}\label{eq:stage-consistency-amgm-bound}
    \frac{1-R}{\prod_{\ell=1}^{s-1}(1-x_j^\ell)}
    \geq
    \auxfun_{k-s}(X_j^s)\left(k(1-R)^{1/k}-k+2-x_j^s\right).
\end{equation}
Substituting \eqref{eq:stage-consistency-amgm-bound} into
\eqref{eq:stage-consistency-main-expression} gives
\begin{align*}
&x_j^s\left(1+P_j^{s-1}+(x_j^s-1)a_j^s(x_j^s)-C_k(R)\right) \\
&\geq
\auxfun_{k-s}(X_j^s)(k(1-R)^{1/k}-k+2-x_j^s)
-\auxfun_{k-s}(X_j^s)
+x_j^s\left(\auxfun_{k-s}(X_j^s)+k-k(1-R)^{1/k}-1\right) \\
&=
\left(\auxfun_{k-s}(X_j^s)-x_j^s\right)(k(1-R)^{1/k}-k+1).
\end{align*}
Finally, recall that $R\in[0,R_k]$ and $R_k=1-(1-1/k)^k$. Hence
$(1-R)^{1/k}\geq1-1/k$, so $k(1-R)^{1/k}-k+1\geq 0$. Moreover,
$\auxfun_{k-s}(X_j^s)\geq X_j^s\geq x_j^s$ by \Cref{lem:g-props}(a).
Therefore 
\[
\left(\auxfun_{k-s}(X_j^s)-x_j^s\right)(k(1-R)^{1/k}-k+1) \geq 0 \ .
\]
Since $x_j^s>0$, we have
\[
1+P_j^{s-1}+(x_j^s-1)a_j^s(x_j^s)-C_k(R) \geq 0,
\]
which proves
\eqref{eq:stage-consistency-suffices}, and hence the lemma.
\end{proof}

\lemmaRobustnessTwo*

\begin{proof}
By \Cref{lem:g-invariant} with $s=k-1$,
\[
    1-R+P_j^{k-1} \geq \auxfun_1(X_j^{k-1})=X_j^{k-1}.
\]
Rearranging gives
\[
    P_j^{k-1}+1-X_j^{k-1}\geq R,
\]
which is exactly \eqref{eq:k-rob-k}.
\end{proof}

\section{Omitted Proofs from \Cref{sec:online} (Online Setting)}\label{app:online}

We now give the full proof of \Cref{thm:online}.  Throughout this appendix section we use the online penalty functions and notation defined in \Cref{sec:online}.

\subsection{The Exponential Safety Curve}

We first give elementary inequalities that replace \Cref{lem:g-props} for the online setting.

\begin{lemma}\label{lem:online-exp-props}
The following statements hold.
\begin{enumerate}[label=(\alph*),nosep]
    \item For every $z\in[0,1]$, it holds that $z \le \auxfun_\infty(z)\le 1$.
    \item For every $z\in[0,1]$ and $y\in[0,1-z]$, it holds that
    \(
        \auxfun_\infty(z)\ge (1-y) \cdot \auxfun_\infty(z+y).
    \)
    \item For every $z\in[0,1]$ and $y\in[0,1-z]$, it holds that
    \(
        \auxfun_\infty(z)+y\ge \auxfun_\infty(z+y).
    \)
\end{enumerate}
\end{lemma}

\begin{proof}
The upper bound in (a) is immediate. For the lower bound, the function $z\mapsto \exp(z-1)-z$ is equal to $0$ at $z=1$ and has derivative $\exp(z-1)-1\le 0$ on $[0,1]$, hence $\exp(z-1)-z\ge 0$ for all $z\in[0,1]$.

For (b), when $y<1$, it is enough to show $1/(1-y)\ge \exp(y)$, which follows from $-\ln(1-y)\ge y$. The only remaining endpoint is $z=0$, $y=1$, where the claimed product-form inequality is immediate because its right-hand side is zero. For (c), define
\[
    L(y):=\exp(z-1)+y-\exp(z+y-1).
\]
Then $L(0)=0$ and $L'(y)=1-\exp(z+y-1) \ge 0$ because $z+y\le 1$. Hence $L(y)\ge 0$.
\end{proof}

\begin{lemma}\label{lem:online-reserve-invariant}
For every $j \in S$ and every time $t\in\{0,1,\ldots,T\}$, we have
\(
    1-R+P_j^t \ge \exp(X_j^t-1) .
\)
\end{lemma}

\begin{proof}
We prove the claim by induction on $t$. For $t=0$, it is exactly $1-R\ge \exp(-1)$.
Assume the invariant holds before time $t$.

If $j\notin A_t$, then $f_j^t=\frob_j^t$. If $\frob_j^t(x_j^t)=1$, then by the induction hypothesis and \Cref{lem:online-exp-props}(c),
\[
    1-R+P_j^t
    =1-R+P_j^{t-1}+x_j^t 
    \ge \exp(X_j^{t-1}-1)+x_j^t 
    \ge \exp(X_j^t-1).
\]
If $\frob_j^t(x_j^t)<1$, then $x_j^t<1$ and the cap in the definition of $\frob_j^t$ is inactive. Therefore
\[
\begin{aligned}
    1-R+P_j^t
    &=1-R+P_j^{t-1}+x_j^t\frac{1-R+P_j^{t-1}}{1-x_j^t} \\
    &=\frac{1-R+P_j^{t-1}}{1-x_j^t} 
    \ge \frac{\exp(X_j^{t-1}-1)}{1-x_j^t} 
    \ge \exp(X_j^t-1),
\end{aligned}
\]
where the last step uses \Cref{lem:online-exp-props}(b).

If $j\in A_t$, then $f_j^t=\fadv_j^t$. If $x_j^t=0$, then $P_j^t=P_j^{t-1}$ and $X_j^t=X_j^{t-1}$, so the induction hypothesis gives
\[
    1-R+P_j^t=1-R+P_j^{t-1}\ge \exp(X_j^{t-1}-1)=\exp(X_j^t-1).
\]
If $x_j^t>0$ and $\fadv_j^t(x_j^t)=0$, then the definition of $\fadv_j^t$ gives $1-R+P_j^{t-1}\ge \exp(X_j^t-1)$, so the invariant remains true. Finally, if $x_j^t>0$ and $\fadv_j^t(x_j^t)>0$, then the definition of $\fadv_j^t$ gives equality:
\[
    1-R+P_j^t
    =1-R+P_j^{t-1}+x_j^t \cdot \fadv_j^t(x_j^t)
    =\exp(X_j^t-1).
\]
This completes the induction.
\end{proof}

\begin{lemma}\label{lem:online-penalties-valid}
For every time $t$ and supply vertex $j$, the penalty function $f_j^t$ is non-decreasing and takes values in $[0,1]$ on $[0,1-X_j^{t-1}]$.
\end{lemma}

\begin{proof}
The non-predicted branch is immediate from the definition and \Cref{lem:online-reserve-invariant}: the numerator $1-R+P_j^{t-1}$ is nonnegative, and $z\mapsto (1-R+P_j^{t-1})/(1-z)$ is non-decreasing on $[0,1)$.

Now consider the predicted branch. Fix $t$ and $j$. For $z>0$, define
\[
    h(z):=\frac{R-1+\exp(X_j^{t-1}+z-1)-P_j^{t-1}}{z}.
\]
Then $\fadv_j^t(z)=\max\{0,h(z)\}$. Computing the derivative of $h$ gives
\[
    h'(z)=\frac{1-R+P_j^{t-1}-\exp(X_j^{t-1}+z-1)(1-z)}{z^2}.
\]
The function $z\mapsto \exp(X_j^{t-1}+z-1)(1-z)$ has derivative $-z\exp(X_j^{t-1}+z-1)\le0$, so it is at most its value at $0$, namely $\exp(X_j^{t-1}-1)\le 1-R+P_j^{t-1}$. Hence $h'(z)\ge0$, and $\fadv_j^t$ is non-decreasing on $(0,1-X_j^{t-1}]$; defining the endpoint by the right limit preserves monotonicity on the closed interval and shows that the limit exists.

The lower bound $\fadv_j^t(z)\ge0$ is by definition. For the upper bound, \Cref{lem:online-reserve-invariant} and \Cref{lem:online-exp-props}(c) imply
\[
    1-R+P_j^{t-1}+z\ge \exp(X_j^{t-1}-1)+z\ge \exp(X_j^{t-1}+z-1).
\]
Thus $h(z)\le1$ for every $z>0$, and therefore $\fadv_j^t(z)\le1$. The endpoint $z=0$ follows by taking the right limit.
\end{proof}

\subsection{Dual Fitting}

Since \Cref{lem:online-penalties-valid} gives non-decreasing penalties in $[0,1]$, the convex program in each stage is well-defined by the same argument as \Cref{lem:conditional-convexity}. The KKT conditions of \Cref{lem:kkt} also apply to every time $t$.
For every time $t$, let $(\lambda_i^t)_{i\in D_t}$ be the KKT multipliers. Define dual variables for $(\DP)$ as follows:
\begin{itemize}
    \item For every demand vertex $i\in D_t$, define $\alpha_i:=\lambda_i^t$.
    \item For every supply vertex $j \in S$, define $\beta_j:=\sum_{t=1}^T\left(w_jx_j^t-\sum_{i\in D_t}\lambda_i^t x_{ij}^t\right)$.
\end{itemize}
Analogously to the proof of \Cref{lem:dual_objective}, the dual objective equals the algorithm's objective value:
\[
    \sum_{i\in D}\alpha_i+\sum_{j\in S}\beta_j=\ALG.
\]
We next give the online analogue of \Cref{lem:intermed_feasibility}. Let $q_j$ be the first time at which $X_j^{q_j}=1$, and set $q_j=\infty$ if $j$ is never saturated.

\begin{lemma}\label{lem:online-dual-certificates}
For every $j\in S$, $\beta_j\ge w_jP_j^T$. Moreover, for every edge $(i,j)\in E_t$, one of the following cases applies:
\begin{enumerate}[label=(\alph*),nosep]
    \item If $X_j^t<1$, then
    \(
        \alpha_i+\beta_j\ge w_j(P_j^T+1-f_j^t(x_j^t)).
    \)
    \item If $q_j<t$, then
    \(
        \alpha_i+\beta_j\ge w_jP_j^{q_j}.
    \)
    \item If $q_j=t$, then
    \(
        \alpha_i+\beta_j\ge w_j(P_j^{t-1}+x_j^t).
    \)
\end{enumerate}
\end{lemma}

\begin{proof}
The proof is the same as the proof of \Cref{lem:intermed_feasibility}, with $T$ in place of the known final stage $k$. We include the details for completeness.

By the KKT conditions, whenever $x_{ij}^t>0$ we have $\lambda_i^t\le w_j(1-f_j^t(x_j^t))$. Hence
\[
    \beta_j
    =\sum_{t=1}^T\left(w_jx_j^t-\sum_{i\in D_t}\lambda_i^t x_{ij}^t\right)
    \ge \sum_{t=1}^T w_jx_j^t f_j^t(x_j^t)
    =w_jP_j^T.
\]
If $X_j^t<1$, then the capacity constraint of $j$ has slack after time $t$, and the KKT condition gives $\alpha_i=\lambda_i^t\ge w_j(1-f_j^t(x_j^t))$. Together with $\beta_j\ge w_jP_j^T$, this proves (a).

If $q_j<t$, then no load is assigned to $j$ after $q_j$, so $P_j^T=P_j^{q_j}$. Since $\alpha_i\ge0$, the lower bound on $\beta_j$ proves (b).

Finally suppose $q_j=t$. The same KKT calculation as in \Cref{lem:intermed_feasibility} shows that the contribution of time $t$ to $\beta_j$, together with $\alpha_i$, is at least $w_jx_j^t$ plus the earlier penalty mass $w_jP_j^{t-1}$. Thus (c) follows.
\end{proof}

\subsection{Consistency and Robustness Certificates}

The remaining work is to prove the consistency and robustness certificates. We first isolate the basic calculation that is used whenever a vertex has only seen non-predicted stages and the non-predicted penalty has not yet reached its cap.
This is the analogue to \Cref{claim:uncapped-nonpredicted-product}.

\begin{lemma}[Uncapped non-predicted prefixes]\label{lem:online-uncapped-prefix}
Fix a supply vertex $j\in S$ and an index $h\in\{0,1,\ldots,T\}$ such that $j\notin A_d$ for every $d\in[h]$.
If $\frob_j^d(x_j^d)<1$ for every $d\in[h]$, then
\[
    1-R+P_j^h=
    \frac{1-R}{\prod_{\ell=1}^{h}(1-x_j^\ell)}.
\]
Moreover, if $t$ is the first index in $[h]$ such that $\frob_j^t(x_j^t)=1$, then for every $d\in\{t,\ldots,h\}$, it holds that $\frob_j^d(x_j^d)=1$ and $P_j^d-X_j^d=P_j^{t-1}-X_j^{t-1}$.
\end{lemma}

\begin{proof}
For $h=0$, the product is empty and the identity is $1-R=1-R$. Now suppose that the cap is inactive in some non-predicted time $d$. Then $x_j^d<1$ and
\[
    \frob_j^d(x_j^d)=\frac{1-R+P_j^{d-1}}{1-x_j^d}.
\]
Substituting this into the definition $P_j^d=P_j^{d-1}+x_j^d\frob_j^d(x_j^d)$ gives
\[
    1-R+P_j^d
    =1-R+P_j^{d-1}
    +x_j^d\frac{1-R+P_j^{d-1}}{1-x_j^d}
    =\frac{1-R+P_j^{d-1}}{1-x_j^d}.
\]
Iterating this equality over the uncapped times $1,\ldots,h$ gives the product formula.

Now let $t$ be the first capped time in $[h]$. If $x_j^t<1$, then the definition of the cap gives
\[
    \frac{1-R+P_j^{t-1}}{1-x_j^t}\ge1,
\]
which is equivalent to $1-R+P_j^{t-1}\ge1-x_j^t$. Hence $1-R+P_j^t=1-R+P_j^{t-1}+x_j^t\ge1$. If $x_j^t=1$, then $1-R+P_j^t=1-R+P_j^{t-1}+1\ge1$ as well.

Once $1-R+P_j^d\ge1$ holds after some time $d$, the next non-predicted penalty is also capped: if $x_j^{d+1}<1$, then
\[
    \frac{1-R+P_j^d}{1-x_j^{d+1}}\ge1,
\]
and if $x_j^{d+1}=1$, the definition gives $\frob_j^{d+1}(x_j^{d+1})=1$. Thus all later non-predicted penalties up to time $h$ are capped. During every capped time $d$, the increment of $P_j^d-X_j^d$ is
\[
    x_j^d \cdot \frob_j^d(x_j^d)-x_j^d=0,
\]
so $P_j^d-X_j^d$ remains equal to its value just before the first cap.
\end{proof}

The next lemma is the analogue of \Cref{Lemma:feasibility_consistency1}.

\begin{lemma}\label{lem:online-prefix-consistency}
For every time $s$, every $j\in A_s$, and every $p\in\{0,1,\ldots,s-1\}$,
we have $P_j^p+1-X_j^p\ge C_\infty(R)$.
\end{lemma}

\begin{proof}
Fix $s$, $j\in A_s$, and $p<s$. Since the predicted sets are disjoint, vertex $j$ is not predicted in any time $d\le p$, so $f_j^d=\frob_j^d$ for all $d\le p$. Let $t$ be the first index in $\{1,\ldots,p\}$ such that $\frob_j^t(x_j^t)=1$; if no such index exists, set $t=p+1$.

The choice of $t$ and \Cref{lem:online-uncapped-prefix} give
\begin{equation}	
    1-R+P_j^{t-1}
    =\frac{1-R}{\prod_{\ell=1}^{t-1}(1-x_j^\ell)}. \label{eq:online-prefix-consistency1}
\end{equation}

If $t=p+1$, the identity $P_j^p-X_j^p=P_j^{t-1}-X_j^{t-1}$ is immediate. If $t\le p$, the second part of \Cref{lem:online-uncapped-prefix} shows that the same identity holds because $P_j^d-X_j^d$ does not change after the first cap. Therefore
\[
    P_j^p+1-X_j^p=P_j^{t-1}+1-X_j^{t-1}.
\]
Using \eqref{eq:online-prefix-consistency1} and $X_j^{t-1}=\sum_{\ell=1}^{t-1}x_j^\ell$, we get
\[
\begin{aligned}
    P_j^p+1-X_j^p
    &=\frac{1-R}{\prod_{\ell=1}^{t-1}(1-x_j^\ell)}-(1-R)+1-
      \sum_{\ell=1}^{t-1}x_j^\ell \\
    &=\frac{1-R}{\prod_{\ell=1}^{t-1}(1-x_j^\ell)}+R-
      \sum_{\ell=1}^{t-1}x_j^\ell \\
    &=\frac{1-R}{\prod_{\ell=1}^{t-1}(1-x_j^\ell)}+
      \sum_{\ell=1}^{t-1}(1-x_j^\ell)-(t-1)+R.
\end{aligned}
\]
The following $t$ positive numbers have product $1-R$:
\[
    1-x_j^1,\ldots,1-x_j^{t-1},
    \frac{1-R}{\prod_{\ell=1}^{t-1}(1-x_j^\ell)}.
\]
Applying $\ln a\le a-1$ to each of these numbers gives
\[
\begin{aligned}
    \ln(1-R)
    &\le
    \sum_{\ell=1}^{t-1}\bigl((1-x_j^\ell)-1\bigr)
    +\left(\frac{1-R}{\prod_{\ell=1}^{t-1}(1-x_j^\ell)}-1\right) \\
    &=\frac{1-R}{\prod_{\ell=1}^{t-1}(1-x_j^\ell)}+
      \sum_{\ell=1}^{t-1}(1-x_j^\ell)-t.
\end{aligned}
\]
Combining this with the above gives
\[
    P_j^p+1-X_j^p
    \ge 1+R+\ln(1-R)
    =C_\infty(R).
\]
This completes the proof.
\end{proof}

\begin{lemma}\label{lem:online-current-consistency}
For every time $s$ and every $j\in A_s$ with $X_j^s<1$,
it holds that
\(
    P_j^s+1-f_j^s(x_j^s)\ge C_{\infty}(R).
\)
Moreover, this implies that
\(
    P_j^T+1-f_j^s(x_j^s)\ge C_{\infty}(R).
\)
\end{lemma}

\begin{proof}
Fix $s$ and $j\in A_s$. Since $f_j^s=\fadv_j^s$, it suffices to prove
\[
    1+P_j^{s-1}+(x_j^s-1)\fadv_j^s(x_j^s)\ge C_\infty(R),
\]
because the left-hand side is exactly $P_j^s+1-\fadv_j^s(x_j^s)$.

First suppose $x_j^s=0$. Recall that $\fadv_j^s(0)$ is defined by the right limit at zero. Unlike its counterpart in the $k$-stage case, this limit can be positive.
If $\fadv_j^s(0)=0$, then the expression is $1+P_j^{s-1}\ge1\ge C_\infty(R)$. It remains to consider the endpoint case $\fadv_j^s(0)>0$. This can only happen if
$1-R+P_j^{s-1}=\exp(X_j^{s-1}-1)$,
and then $\fadv_j^s(0)=\exp(X_j^{s-1}-1)$.

No earlier non-predicted penalty can have been capped. Indeed, after a cap we would have $P_j^{s-1}\ge R$ by \Cref{lem:online-uncapped-prefix}. Together with $1-R+P_j^{s-1}=\exp(X_j^{s-1}-1)\le1$, this would imply $P_j^{s-1}=R$ and $X_j^{s-1}=1$, contradicting $X_j^s<1$.
Thus \Cref{lem:online-uncapped-prefix} gives
\begin{equation}
    1-R+P_j^{s-1}
    =\frac{1-R}{\prod_{\ell=1}^{s-1}(1-x_j^\ell)}.
    \label{eq:online-current-consistency1}
\end{equation}
Taking logarithms in the equality $1-R+P_j^{s-1}=\exp(X_j^{s-1}-1)$ and substituting \eqref{eq:online-current-consistency1} gives
\[
\begin{aligned}
    0
    &=\ln\left(\frac{1-R+P_j^{s-1}}{\exp(X_j^{s-1}-1)}\right) \\
    &=\ln\left(\frac{1-R}{\exp(X_j^{s-1}-1) \prod_{\ell=1}^{s-1}(1-x_j^\ell)}\right) \\
    &=\ln(1-R)+1+
      \sum_{\ell=1}^{s-1}\bigl(-\ln(1-x_j^\ell)-x_j^\ell\bigr).
\end{aligned}
\]
All terms on the right-hand side are nonnegative, since $\ln(1-R)+1\ge0$ and $-\ln(1-y)-y\ge0$ for $y\in[0,1)$. Hence equality can hold only when $R=1-\nicefrac1e$. For this value of $R$, we have $C_\infty(R)=R$, and therefore
\[
    1+P_j^{s-1}-\fadv_j^s(0)
    =1+P_j^{s-1}-\exp(X_j^{s-1}-1)
    =R
    =C_\infty(R).
\]
This proves the case $x_j^s=0$.

Now assume $x_j^s>0$. If $\fadv_j^s(x_j^s)=0$, then the desired inequality is again $1+P_j^{s-1}\ge1\ge C_\infty(R)$. It remains to consider the case $\fadv_j^s(x_j^s)>0$. In this case no earlier non-predicted penalty of $j$ can have been capped. Indeed, if some cap occurred before $s$, then $P_j^{s-1}\ge R$, and the numerator defining $\fadv_j^s(x_j^s)$ would be at most $\exp(X_j^s-1)-1\le0$, a contradiction.

Therefore \Cref{lem:online-uncapped-prefix} gives
\begin{equation}	
    1-R+P_j^{s-1}
    =\frac{1-R}{\prod_{\ell=1}^{s-1}(1-x_j^\ell)}.
    \label{eq:online-current-consistency2}
\end{equation}

Since $\fadv_j^s(x_j^s)>0$, the definition of the predicted penalty gives
\[
    x_j^s \cdot \fadv_j^s(x_j^s)
    =\exp(X_j^s-1)-(1-R+P_j^{s-1}).
\]
Multiplying the desired inequality by $x_j^s>0$ and substituting this identity shows that it is equivalent to
\begin{equation}	
    1-R+P_j^{s-1}-\exp(X_j^s-1)
    +x_j^s\bigl(\exp(X_j^s-1)-1-\ln(1-R)\bigr)
    \ge0. \label{eq:online-current-consistency3}
\end{equation}
We lower-bound the first term $1-R+P_j^{s-1}$ in terms of $\exp(X_j^s-1)$. Using \eqref{eq:online-current-consistency2},
\[
\begin{aligned}
    \ln\left(\frac{1-R+P_j^{s-1}}{\exp(X_j^s-1)}\right)
    &=\ln(1-R)-\sum_{\ell=1}^{s-1}\ln(1-x_j^\ell)-X_j^s+1 \\
    &=\ln(1-R)+1-x_j^s+
      \sum_{\ell=1}^{s-1}\bigl(-\ln(1-x_j^\ell)-x_j^\ell\bigr) \\
    &\ge \ln(1-R)+1-x_j^s.
\end{aligned}
\]
Since $\exp(y)\ge1+y$ for every $y$, this implies
\[
    1-R+P_j^{s-1}
    \ge \exp(X_j^s-1)\exp(\ln(1-R)+1-x_j^s)
    \ge \exp(X_j^s-1)\bigl(2+\ln(1-R)-x_j^s\bigr).
\]
Substituting this lower bound into \eqref{eq:online-current-consistency3} gives
\[
\begin{aligned}
    &1-R+P_j^{s-1}-\exp(X_j^s-1)
    +x_j^s\bigl(\exp(X_j^s-1)-1-\ln(1-R)\bigr) \\
    &\qquad\ge
    \bigl(\exp(X_j^s-1)-x_j^s\bigr)\bigl(1+\ln(1-R)\bigr).
\end{aligned}
\]
By \Cref{lem:online-exp-props}(a), we have $\exp(X_j^s-1)\ge X_j^s\ge x_j^s$, and $1+\ln(1-R)\ge0$. This proves the first claim. The second claim follows because $P_j^T\ge P_j^s$.
\end{proof}

\begin{lemma}[Robustness]\label{lem:online-robustness-certificates}
For every edge $(i,j)\in E_t$, the dual solution satisfies
$\alpha_i+\beta_j\ge Rw_j$.

\end{lemma}

\begin{proof}
We use the cases of \Cref{lem:online-dual-certificates}. First suppose $X_j^t<1$. If $j\in A_t$, then \Cref{lem:online-current-consistency} gives
\[
    P_j^T+1-f_j^t(x_j^t)\ge C_\infty(R)\ge R.
\]
If $j\notin A_t$, then $f_j^t=\frob_j^t$. We claim that
\begin{equation}
    1+P_j^{t-1}+(x_j^t-1)\frob_j^t(x_j^t)\ge R. \label{eq:online-robustness-certificates}
\end{equation}
This is immediate if $x_j^t=1$. If $x_j^t<1$ and the cap is inactive, then \eqref{eq:online-robustness-certificates} holds with equality by the definition of $\frob_j^t$. If the cap is active, then the definition of the cap gives $1-R+P_j^{t-1}\ge1-x_j^t$, which is equivalent to \eqref{eq:online-robustness-certificates}. Hence
\[
    P_j^T+1-\frob_j^t(x_j^t)
    \ge P_j^t+1-\frob_j^t(x_j^t)
    =1+P_j^{t-1}+(x_j^t-1)\frob_j^t(x_j^t)
    \ge R.
\]
Together with \Cref{lem:online-dual-certificates}(a), this proves the desired inequality in the case $X_j^t<1$.

If $q_j<t$, then $X_j^{q_j}=1$. By \Cref{lem:online-reserve-invariant}, $1-R+P_j^{q_j}\ge \exp(0)=1$, so $P_j^{q_j}\ge R$. \Cref{lem:online-dual-certificates}(b) gives the result.

Finally, if $q_j=t$, then $x_j^t=1-X_j^{t-1}$. By \Cref{lem:online-reserve-invariant} and \Cref{lem:online-exp-props}(a),
\[
    1-R+P_j^{t-1}\ge \exp(X_j^{t-1}-1)\ge X_j^{t-1}.
\]
Thus $P_j^{t-1}+x_j^t=P_j^{t-1}+1-X_j^{t-1}\ge R$, and \Cref{lem:online-dual-certificates}(c) finishes the proof.
\end{proof}

\begin{lemma}[Consistency certificates]\label{lem:online-consistency-certificates}
Let $\widehat M$ be the integral prediction.
For every $(i,j)$ in $\widehat M$, the dual solution satisfies
\(
    \alpha_i+\beta_j\ge C_\infty(R)w_j.
\)
\end{lemma}

\begin{proof}
Let $(i,j) \in \widehat M$ be a predicted edge at time $s$, so $j\in A_s$. Again use the cases of \Cref{lem:online-dual-certificates}. If $X_j^s<1$, then \Cref{lem:online-current-consistency} gives
\[
    \alpha_i+\beta_j
    \ge w_j\left(P_j^T+1-f_j^s(x_j^s)\right)
    \ge C_\infty(R)w_j.
\]
If $q_j=s$, then $x_j^s=1-X_j^{s-1}$, and \Cref{lem:online-prefix-consistency} with $p=s-1$ gives
\[
    P_j^{s-1}+x_j^s=P_j^{s-1}+1-X_j^{s-1}\ge C_\infty(R).
\]
The claim follows from \Cref{lem:online-dual-certificates}(c). If $q_j<s$, then \Cref{lem:online-prefix-consistency} with $p=q_j$ gives
\[
    P_j^{q_j}=P_j^{q_j}+1-X_j^{q_j}\ge C_\infty(R),
\]
and \Cref{lem:online-dual-certificates}(b) applies.
\end{proof}

\begin{proof}[Proof of \Cref{thm:online}]
Consider any finite input sequence of length $T$. By \Cref{lem:online-robustness-certificates}, $(\alpha/R,\beta/R)$ is feasible for the offline dual whenever $R>0$. Hence weak duality and the equality between the dual objective and $\ALG$ give $\ALG\ge R \cdot \OPT$. The case $R=0$ is trivial.

For consistency, let $
\widehat M$ be the predicted matching. Since all dual variables are nonnegative, \Cref{lem:online-consistency-certificates} implies
\[
    \ALG
    =\sum_{i\in D}\alpha_i+\sum_{j\in S}\beta_j
    \ge \sum_{(i,j)\in 
    \widehat M}(\alpha_i+\beta_j)
    \ge C_\infty(R)\sum_{(i,j)\in 
    \widehat M}w_j
    =C_\infty(R)\cdot \advice.
\]
The extension to online fractional AdWords with fractional predictions follows by applying the virtual-advertiser construction of \Cref{app:AdWords} with the limiting online penalty functions. The dual fitting, load inequalities, and safety invariants are per virtual advertiser, so the proof above applies after the same aggregation and preservation arguments used in \Cref{lem:virtual-preserves-benchmarks} and \Cref{lem:lazy-virtual-implementation}.
\end{proof}

\section{Fractional AdWords with Fractional Predictions}
\label{app:AdWords}

We now extend the algorithm and analysis to fractional AdWords with fractional predictions.

\subsection{Virtual Advertiser Construction}

For the analysis, it is convenient to describe a full virtual instance that is based on the predictions of all stages.
We will later show that the virtual instance can be constructed lazily online, so we do not need to know the predictions of future stages in advance.

For every original advertiser $j$ and every non-final stage
$\tau\in[k-1]$, create a virtual advertiser $(j,\tau)$ with budget
\[
\widehat B_{(j,\tau)} := \hat x^\tau_j B_j,
\qquad\text{where}\qquad
\hat x_j^\tau := \frac{1}{B_j}\sum_{i\in D_\tau:(i,j)\in E_\tau} b_{ij}\hx_{ij}^\tau .
\]
In addition, create a residual-capacity virtual advertiser $(j,k)$ with budget
\[
\widehat B_{(j,k)}
:=
\hat\rho_j B_j,
\qquad\text{where}\qquad
\hat\rho_j:=1-\sum_{\tau=1}^{k-1}\hat x^\tau_j .
\]
This copy represents the full budget of $j$ left after the predicted prefix;
it is not required to equal the amount spent by the final-stage residual prediction.
If $\hB_h=0$, we omit virtual advertiser $h$. 
Every original edge $(i,j)$ is copied to every virtual advertiser $(j,\tau)$, with the same bid $b_{i,(j,\tau)}:=b_{ij}$.
Let $\mathcal H$ denote the set of virtual advertisers, and let $E_s^{\mathrm{virt}}$ denote the copied edge set in stage $s$.

The fractional prediction is lifted as follows. For every non-final stage $s<k$,
the predicted allocation $\hx_{ij}^s$ on an original edge $(i,j)\in E_s$ is placed
on the virtual edge $(i,(j,s))$. This exactly fills the virtual advertiser $(j,s)$ under the prediction, since
\[
    \sum_{i\in D_s:(i,j)\in E_s} b_{ij}\hx_{ij}^s
    =
    \hx_j^s B_j
    =
    \hB_{(j,s)} .
\]
The final residual allocation $\hx^k$ is placed on the residual-capacity edge
$(i,(j,k))$. This is feasible because $\hx^k$ is chosen in the residual instance
after the predicted prefix, and hence
\[
    \sum_{i\in D_k:(i,j)\in E_k} b_{ij}\hx_{ij}^k
    \le
    \left(1-\sum_{s=1}^{k-1}\hx_j^s\right)B_j
    = \hB_{(j,k)} .
\]

For vertex-weighted fractional bipartite matching, this is the same construction
with $b_{ij}=B_j=w_j$.  Equivalently, one may ignore budgets and split each
supply vertex $j$ into virtual copies of capacities $\hx_j^1,\ldots,\hx_j^{k-1}$
and residual capacity $1-\sum_{s=1}^{k-1}\hx_j^s$, all carrying value $w_j$ per
unit. The aggregation and splitting arguments below then become the usual
capacity-preservation arguments for fractional matchings.

\begin{lemma}
\label{lem:virtual-preserves-benchmarks}
Given a fractional prediction $\hx$,
let $\opt_{\mathrm{virt}}$ and $\advice_{\mathrm{virt}}$ denote the optimal objective value and predicted objective value of the virtual instance. Then
$\opt=\opt_{\mathrm{virt}}$ and $\advice=\advice_{\mathrm{virt}}$.
\end{lemma}

\begin{proof}
First, any feasible allocation $z$ in the virtual instance maps to a feasible allocation $x$ in the original instance by aggregating over all virtual copies of each advertiser:
\[
    x_{ij}:=\sum_{\tau=1}^k z_{i,(j,\tau)} .
\]
The demand constraints are preserved because the virtual demand constraint already sums over all virtual advertisers. Moreover, for every original advertiser $j$,
\[
    \sum_{i \in D} b_{ij}x_{ij}
    =
    \sum_{\tau=1}^k\sum_{i \in D} b_{i,(j,\tau)}z_{i,(j,\tau)}
    \le
    \sum_{\tau=1}^k \hB_{(j,\tau)}
    =
    B_j .
\]
The objective value is unchanged by aggregation. Hence $\opt\ge \opt_{\mathrm{virt}}$.

Conversely, take any feasible allocation $x$ in the original instance. For every original advertiser $j$, define the spend assigned to edge $(i,j)$ as $p_{ij}:=b_{ij}x_{ij}$. Since
\[
    \sum_{i \in D} p_{ij}\le B_j=\sum_{\tau=1}^k \hB_{(j,\tau)},
\]
we can split these spends among the virtual advertisers of $j$: choose nonnegative numbers $p_{ij}^\tau$ such that $\sum_{\tau=1}^k p_{ij}^\tau=p_{ij}$ and $\sum_{i \in D} p_{ij}^\tau\le \hB_{(j,\tau)}$ for every $\tau\in[k]$.
This can be done greedily by filling the virtual budgets of $j$. Now set
\[
    z_{i,(j,\tau)}:=\frac{p_{ij}^\tau}{b_{ij}} .
\]
The demand constraints are preserved because
\[
    \sum_{j,\tau} z_{i,(j,\tau)}
    =
    \sum_j \frac{1}{b_{ij}}\sum_\tau p_{ij}^\tau
    =
    \sum_j x_{ij}
    \le 1,
\]
and the virtual budget constraints hold by construction. The objective value is unchanged. Hence $\opt_{\mathrm{virt}}\ge \opt$, and therefore $\opt_{\mathrm{virt}}=\opt$.

The prediction benchmark is preserved as well. After the lifted non-final
prediction is followed, each non-final copy $(j,s)$ with $s<k$ is saturated, and
therefore the only remaining virtual capacity corresponding to $j$ is the
residual-capacity copy $(j,k)$. Its budget is exactly
\[
    \left(1-\sum_{s=1}^{k-1}\hx_j^s\right)B_j,
\]
which is the budget left for $j$ in the original residual instance. Since the
final-stage original edges are copied to $(j,k)$ with the same bids, the final
residual LP in the virtual instance is identical to the final residual LP in the
original instance. Therefore their optimal residual completions have the same
value, and $\advice_{\mathrm{virt}}=\advice$.
\end{proof}

\subsection{Algorithm on the Virtual Instance}

Fix a virtual advertiser $h=(j,\tau)\in\mathcal H$. 
For every stage $s\in[k]$, define the normalized stage-$s$ load of
virtual advertiser $h$ by
\[
    x_h^s
    :=
    \frac{1}{\widehat B_h}
    \sum_{i\in D_s:(i,h)\in E_s^{\mathrm{virt}}} b_{ih}x_{ih}^s .
\]
For $q \in [k]$, define
$X_h^q:=\sum_{s=1}^q x_h^s$
and
$P_h^q:=\sum_{s=1}^q x_h^sf_h^s(x_h^s)$.
We also set $X_h^0=P_h^0=0$.

In a non-final stage $s$, the virtual advertiser $h=(j,\tau)$ is predicted if and only if $\tau=s$. Its penalty is defined exactly as in \Cref{sec:penalty-functions}. That is, on the feasible interval $z\in[0,1-X_h^{s-1}]$, define for all $z \in [0,1)$
\[
    r_h^s(z):=\min\left\{1,\frac{1-R+P_h^{s-1}}{1-z}\right\}
\]
and $r_h^s(1):=1$
and for all $z \in (0,1]$
\[
    a_h^s(z):=\max\left\{0,\frac{R-1+\auxfun_{k-s}(X_h^{s-1}+z)-P_h^{s-1}}{z}\right\}
\]
and $a_h^s(0):=0$.
Finally, set for all $z \in [0,1-X_h^{s-1}]$
\[
    f_h^s(z):=
    \begin{cases}
        a_h^s(z), & \text{if } h=(j,\tau) \text{ with } \tau=s,\\
        r_h^s(z), & \text{otherwise.}
    \end{cases}
\]
For the final stage $k$, we set $f^k_h := 0$ for all $h\in\mathcal H$. 

In each stage $s \in [k]$, after $D_s$, $E_s$, and the stage-$s$ prediction $\hx^s$ are revealed, the algorithm solves
\begin{alignat}{3}
    (\CP_s^{\mathrm{Ad}}) \quad \max \quad
    &\sum_{h\in\mathcal H}\hB_h\left(x_h^s-\int_0^{x_h^s}f_h^s(t)\,dt\right) \notag \label{eq:AdWords-stage-cp}\\
    \text{s.t.}\quad
    &\sum_{h:(i,h)\in E_s^{\mathrm{virt}}}x_{ih}^s\le 1 &&\qquad \forall i\in D_s \notag\\
    &x_h^s=\frac{1}{\hB_h}\sum_{i\in D_s:(i,h)\in E_s^{\mathrm{virt}}} b_{ih}x_{ih}^s &&\qquad \forall h\in\mathcal H \notag\\
    &x_h^s\le 1-X_h^{s-1} &&\qquad \forall h\in\mathcal H \notag\\
    &x_{ih}^s\ge 0 &&\qquad \forall (i,h)\in E_s^{\mathrm{virt}} . \notag
\end{alignat}
The equality defining $x_h^s$ is included to make the normalized load explicit. In the KKT derivation below, we use the equivalent formulation obtained by substituting this equality into the objective and the capacity constraint; thus the optimization variables are the edge allocations $x_{ih}^s$.

The allocation returned in the original AdWords instance is obtained by aggregating the virtual allocations over all copies of each original advertiser. This aggregation preserves feasibility and value by the same argument as in \Cref{lem:virtual-preserves-benchmarks}.

\begin{lemma}[AdWords KKT conditions]
\label{lem:AdWords:kkt}
Let $x^s$ be an optimal solution of $(\CP_s^{\mathrm{Ad}})$.
For every stage $s\in[k]$, there exist nonnegative multipliers
$\lambda_i^s$ for $i\in D_s$, $\theta_h^s$ for $h\in\mathcal H$, and
$\gamma_{ih}^s$ for $(i,h)\in E_s^{\mathrm{virt}}$ such that, for every
$(i,h)\in E_s^{\mathrm{virt}}$,
\begin{equation}\label{eq:AdWords-kkt-stationarity}
    b_{ih}\bigl(1-f_h^s(x_h^s)\bigr)
    =
    \lambda_i^s+\frac{b_{ih}}{\widehat B_h}\theta_h^s-\gamma_{ih}^s .
\end{equation}
Moreover, these multipliers satisfy complementary slackness:
\begin{align}
    \lambda_i^s\biggl(1-
        \sum_{h:(i,h)\in E_s^{\mathrm{virt}}}x_{ih}^s
    \biggr)&=0 &&\forall i\in D_s, \label{eq:AdWords-kkt-demand-cs}\\
    \theta_h^s(1-X_h^{s-1}-x_h^s)&=0 &&\forall h\in\mathcal H, \label{eq:AdWords-kkt-budget-cs}\\
    \gamma_{ih}^s x_{ih}^s&=0 &&\forall (i,h)\in E_s^{\mathrm{virt}}. \label{eq:AdWords-kkt-nonneg-cs}
\end{align}
In particular,
\begin{enumerate}[label=(\alph*),nosep]
    \item if $x_{ih}^s>0$, then
    $\lambda_i^s\le b_{ih}\bigl(1-f_h^s(x_h^s)\bigr)$;
    \item if $\sum_{h:(i,h)\in E_s^{\mathrm{virt}}}x_{ih}^s<1$, then
    $\lambda_i^s=0$;
    \item if $X_h^s<1$, then, for every $(i,h)\in E_s^{\mathrm{virt}}$,
    $\lambda_i^s\ge b_{ih}\bigl(1-f_h^s(x_h^s)\bigr)$.
\end{enumerate}
\end{lemma}

\begin{proof}
Fix a stage $s$. We derive the KKT conditions after eliminating the auxiliary
load variables $x_h^s$. For each $h\in\mathcal H$, define
\[
    L_h(x):=\frac{1}{\widehat B_h}
    \sum_{i\in D_s:(i,h)\in E_s^{\mathrm{virt}}} b_{ih}x_{ih} .
\]
The reduced problem is
\begin{align*}
\max \quad
&\sum_{h\in\mathcal H}\widehat B_h
\left(L_h(x)-\int_0^{L_h(x)}f_h^s(t)\,dt\right) \\
\text{s.t.}\quad
&\sum_{h:(i,h)\in E_s^{\mathrm{virt}}}x_{ih}\le 1
    &&\forall i\in D_s,\\
&L_h(x)\le 1-X_h^{s-1}
    &&\forall h\in\mathcal H,\\
&x_{ih}\ge 0
    &&\forall (i,h)\in E_s^{\mathrm{virt}}.
\end{align*}
This problem is equivalent to $(\CP_s^{\mathrm{Ad}})$, and zero-budget virtual
advertisers have been omitted, so $\widehat B_h>0$ for every $h\in\mathcal H$.
By the same argument as in \Cref{lem:penalties-valid}, the penalty functions
on the virtual advertisers are continuous and non-decreasing. Therefore the reduced objective is differentiable and concave,
and its partial derivative with respect to $x_{ih}$, evaluated at the optimum
$x^s$, is
\[
    \widehat B_h\bigl(1-f_h^s(x_h^s)\bigr)\cdot \frac{b_{ih}}{\widehat B_h}
    = b_{ih}\bigl(1-f_h^s(x_h^s)\bigr).
\]
Since the feasible region is a polytope, the first-order optimality conditions
for this concave maximization problem with affine constraints give nonnegative
multipliers $\lambda^s,\theta^s,\gamma^s$ for the demand constraints, normalized
capacity constraints, and nonnegativity constraints, respectively.
Writing the Lagrangian as
\begin{align*}
\mathcal L(x,\lambda,\theta,\gamma)
&=
\sum_{h\in\mathcal H}\widehat B_h
\left(L_h(x)-\int_0^{L_h(x)}f_h^s(t)\,dt\right) \\
&\quad +\sum_{i\in D_s}\lambda_i^s
\left(1-\sum_{h:(i,h)\in E_s^{\mathrm{virt}}}x_{ih}\right) \\
&\quad +\sum_{h\in\mathcal H}\theta_h^s(1-X_h^{s-1}-L_h(x))
    +\sum_{(i,h)\in E_s^{\mathrm{virt}}}\gamma_{ih}^s x_{ih},
\end{align*}
stationarity with respect to $x_{ih}$ gives \eqref{eq:AdWords-kkt-stationarity}.
The complementary slackness conditions give
\eqref{eq:AdWords-kkt-demand-cs}--\eqref{eq:AdWords-kkt-nonneg-cs}, using
$L_h(x^s)=x_h^s$.

It remains only to derive the three simplified consequences. If $x_{ih}^s>0$,
then $\gamma_{ih}^s=0$, so \eqref{eq:AdWords-kkt-stationarity} gives
\[
    b_{ih}\bigl(1-f_h^s(x_h^s)\bigr)
    =\lambda_i^s+\frac{b_{ih}}{\widehat B_h}\theta_h^s
    \ge \lambda_i^s.
\]
If the demand constraint for $i$ is slack, then \eqref{eq:AdWords-kkt-demand-cs}
gives $\lambda_i^s=0$. Finally, if $X_h^s<1$, then
$1-X_h^{s-1}-x_h^s>0$, so \eqref{eq:AdWords-kkt-budget-cs} gives
$\theta_h^s=0$. Thus \eqref{eq:AdWords-kkt-stationarity} gives
\[
    \lambda_i^s
    =b_{ih}\bigl(1-f_h^s(x_h^s)\bigr)+\gamma_{ih}^s
    \ge b_{ih}\bigl(1-f_h^s(x_h^s)\bigr).
\]
\end{proof}

\subsection{Load Inequalities for Virtual Advertisers}

For every virtual advertiser $h$, let $q_h$ denote the first stage $q\in[k]$ such that $X_h^q=1$; if this never happens, set $q_h=\infty$.

\begin{lemma}[Inherited load inequalities]
\label{obs:AdWords-load}
The inequalities from \Cref{Lemma:feasibility_robustness1,Lemma:feasibility_robustness2,Lemma:feasibility_consistency1,Lemma:feasibility_consistency2} apply
to every virtual advertiser $h\in\mathcal H$.
In particular, for every $h\in\mathcal H$ and every $s\in[k-1]$,
\[
    P_h^{k-1}+1-f_h^s(x_h^s)\ge R,
\]
and
\[
    P_h^{k-1}+1-X_h^{k-1}\ge R.
\]
Moreover, if $h=(j,\tau)$ with $\tau < k$, then
\[
    P_h^{k-1}+1-f_h^\tau(x_h^\tau)\ge C_k(R) \ ,
\]
and if $\tau = k$, then
\[
    P_h^{k-1}+1-X_h^{k-1}\ge C_k(R) \ ,
\]
Further, for every $h=(j,\tau)$ and $q \in \{0,\ldots,\tau-1\}$,
\[
    P_h^q+1-X_h^q\ge C_k(R).
\]
\end{lemma}

\begin{proof}
The proofs of the load inequalities in \Cref{Lemma:feasibility_robustness1,Lemma:feasibility_robustness2,Lemma:feasibility_consistency1,Lemma:feasibility_consistency2}
are per advertiser: they use only the normalized load sequence
$(x_h^1,\ldots,x_h^{k-1})$ and the unique stage in which $h$ is predicted.
For the virtual instance, use the designated prediction-stage classes
\[
    A_s^{\mathrm{virt}}
    :=
    \{(j,s): j\in S,\ \widehat B_{(j,s)}>0\},
\]
for all $s \in [k]$.
These classes are pairwise disjoint and cover all virtual advertisers. The
lifted prediction in stage $s$ is supported only on $A_s^{\mathrm{virt}}$; for
$s=k$ the residual-capacity copy may be only partially used, but the load
inequalities do not require it to be filled. Therefore the inequalities of
\Cref{sec:dual_fitting} apply directly.
\end{proof}

\subsection{Dual Fitting on the Virtual Instance}

The offline AdWords LP on the virtual instance can be written as follows.
\begin{alignat*}{3}
    (\LP^{\mathrm{Ad}}) \quad \max\quad &\sum_{(i,h)\in E^{\mathrm{virt}}} b_{ih}z_{ih} \notag \\
    \text{s.t.}\quad
    &\sum_{h:(i,h)\in E^{\mathrm{virt}}} z_{ih}\le 1
        &&\qquad \forall i\in D,\\
    &\sum_{i:(i,h)\in E^{\mathrm{virt}}}\frac{b_{ih}}{\widehat B_h}z_{ih}\le 1
        &&\qquad \forall h\in\mathcal H,\\
    &z_{ih}\ge 0
        &&\qquad \forall (i,h)\in E^{\mathrm{virt}} .
\end{alignat*}
Its dual is the following linear program.
\begin{alignat*}{3}
    (\DP^{\mathrm{Ad}}) \quad  \min\quad &\sum_{i\in D}\alpha_i+\sum_{h\in\mathcal H}\beta_h \notag \\
    \text{s.t.}\quad
    &\alpha_i+\frac{b_{ih}}{\widehat B_h}\beta_h\ge b_{ih}
        &&\qquad \forall (i,h)\in E^{\mathrm{virt}},\\
    &\alpha_i,\beta_h\ge 0
        &&\qquad \forall i\in D,\ h\in\mathcal H .
\end{alignat*}

We next define an allocation of dual variables $(\alpha,\beta)$
of $(\DP^{\mathrm{Ad}})$. For every stage $s \in [k]$, let $(\lambda_i^s)_{i \in D_s}$
be given by \Cref{lem:AdWords:kkt}.
\begin{itemize}
	\item
For every $i\in D_s$, define
\(
    \alpha_i:=\lambda_i^s .
\)
\item
For every virtual advertiser $h\in\mathcal H$, define
\(
    \beta_h
    :=
    \sum_{s=1}^{k}
    \left(
        \widehat B_h x_h^s
        -
        \sum_{i\in D_s}\lambda_i^s x_{ih}^s
    \right).
\)
\end{itemize}

\begin{lemma}[Dual objective value]
\label{lem:AdWords-dual-objective}
The dual objective value of the defined assignment equals the objective value
of the algorithm on the virtual instance, that is, 
$\sum_{i\in D}\alpha_i+\sum_{h\in\mathcal H}\beta_h = \alg_{\mathrm{virt}}$. 
\end{lemma}

\begin{proof}
For every $h\in\mathcal H$, the definition of $\beta_h$ gives
\[
    \beta_h+
    \sum_{s=1}^{k}\sum_{i\in D_s}\alpha_i x_{ih}^s
    =
    \widehat B_h X_h^k .
\]
By \Cref{lem:AdWords:kkt}(b), if the demand constraint of $i$ is slack then
$\alpha_i=0$, while if it is tight then
$\sum_{h:(i,h)\in E_s^{\mathrm{virt}}}x_{ih}^s=1$. Hence, for every
$i\in D_s$,
\[
    \alpha_i
    =
    \alpha_i\sum_{h:(i,h)\in E_s^{\mathrm{virt}}}x_{ih}^s .
\]
Therefore
\begin{align*}
    \sum_{h\in\mathcal H}\beta_h+\sum_{i\in D}\alpha_i
    &=
    \sum_{h\in\mathcal H}\beta_h+
    \sum_{s=1}^{k}\sum_{i\in D_s}\sum_{h:(i,h)\in E_s^{\mathrm{virt}}}\alpha_i x_{ih}^s \\
    &=
    \sum_{h\in\mathcal H}\left(
        \beta_h+
        \sum_{s=1}^{k}\sum_{i\in D_s}\alpha_i x_{ih}^s
    \right) \\
    &=
    \sum_{h\in\mathcal H}\widehat B_h X_h^k
    =
    \ALG_{\mathrm{virt}}.
\end{align*}
\end{proof}

\begin{lemma}[Dual Properties]
\label{lem:AdWords-certificate}
For every $h\in\mathcal H$, $\beta_h\ge \widehat B_h P_h^{k-1}$.
For every edge $(i,h)\in E_s^{\mathrm{virt}}$, one of the following
cases applies:
\begin{enumerate}[label=(\alph*),nosep]
    \item If $X_h^s<1$, then
    \[
        \alpha_i+\frac{b_{ih}}{\widehat B_h}\beta_h
        \ge
        b_{ih}\left(P_h^{k-1}+1-f_h^s(x_h^s)\right).
    \]
    \item If $q_h<s$, then
    \[
        \alpha_i+\frac{b_{ih}}{\widehat B_h}\beta_h
        \ge
        b_{ih}P_h^{q_h}.
    \]
    \item If $q_h=s$, then
    \[
        \alpha_i+\frac{b_{ih}}{\widehat B_h}\beta_h
        \ge
        b_{ih}\left(P_h^{s-1}+x_h^s\right).
    \]
\end{enumerate}
In particular, for every final-stage edge $(i,h)\in E_k^{\mathrm{virt}}$,
it holds that
\[
    \alpha_i+\frac{b_{ih}}{\widehat B_h}\beta_h
    \ge
    b_{ih}\left(P_h^{k-1}+1-X_h^{k-1}\right).
\]
\end{lemma}

\begin{proof}
\Cref{lem:AdWords:kkt}(a) gives
$\lambda_i^s\le b_{ih}(1-f_h^s(x_h^s))$ whenever $x_{ih}^s>0$. Hence

\begin{align*}
    \beta_h
    &=
    \sum_{s=1}^{k}
    \left(
        \widehat B_hx_h^s-\sum_{i\in D_s}\lambda_i^s x_{ih}^s
    \right)  \\
    &\ge
    \sum_{s=1}^{k}
    \left(
        \widehat B_hx_h^s
        -
        \sum_{i\in D_s}b_{ih}(1-f_h^s(x_h^s))x_{ih}^s
    \right) \\
    &=
    \widehat B_h\sum_{s=1}^{k}x_h^sf_h^s(x_h^s)
    =
    \widehat B_hP_h^{k-1},
\end{align*}
because $f_h^k = 0$.

Now fix $(i,h)\in E_s^{\mathrm{virt}}$. If $X_h^s<1$, then the normalized capacity
constraint of $h$ is slack in stage $s$, so \Cref{lem:AdWords:kkt}(c)
gives
\[
    \alpha_i=\lambda_i^s\ge b_{ih}(1-f_h^s(x_h^s)).
\]
Together with $\beta_h\ge \widehat B_hP_h^{k-1}$, this proves case (a).

If $q_h<s$, then $h$ was saturated before stage $s$, so $P_h^{k-1}=P_h^{q_h}$.
Using $\alpha_i\ge0$ and the lower bound on $\beta_h$ proves case (b).

It remains to consider $q_h=s$. Let $\theta_h^s$ be the KKT multiplier
for the normalized capacity constraint of $h$ in stage $s$. Stationarity gives, for
every edge $(r,h)\in E_s^{\mathrm{virt}}$,
\[
    \lambda_r^s
    \ge
    b_{rh}
    \left(
        1-f_h^s(x_h^s)-\frac{\theta_h^s}{\widehat B_h}
    \right),
\]
with equality on every used edge into $h$. Therefore, using
$\sum_{m\in D_s}b_{mh}x_{mh}^s=\widehat B_hx_h^s$, the stage-$s$
contribution to $\beta_h$ is
\[
    \widehat B_hx_h^s f_h^s(x_h^s)+\theta_h^sx_h^s .
\]
Earlier stages contribute at least $\widehat B_hP_h^{s-1}$ and later
stages contribute nonnegatively. Hence

\begin{align*}
    \alpha_i+\frac{b_{ih}}{\widehat B_h}\beta_h
    &\ge
    b_{ih}\left(
        P_h^{s-1}
        +1-f_h^s(x_h^s)
        +x_h^sf_h^s(x_h^s)
        -(1-x_h^s)\frac{\theta_h^s}{\widehat B_h}
    \right).
\end{align*}
Since $q_h=s$, some edge into $h$ is used in stage $s$. For such a used
edge, nonnegativity of the corresponding $\lambda$ implies
\[
    \frac{\theta_h^s}{\widehat B_h}
    \le
    1-f_h^s(x_h^s).
\]
Substituting this bound gives
\[
    \alpha_i+\frac{b_{ih}}{\widehat B_h}\beta_h
    \ge
    b_{ih}(P_h^{s-1}+x_h^s),
\]
which proves case (c).

The final-stage bound follows from the three cases. If $X_h^k<1$, then
case (a) and $f_h^k=0$ give a stronger bound. If $X_h^k=1$ and $q_h<k$,
then $X_h^{k-1}=1$ and case (b) gives
$b_{ih}P_h^{k-1}$. If $q_h=k$, then
$x_h^k=1-X_h^{k-1}$ and case (c) gives the displayed bound.
\end{proof}

\subsection{Robustness and Consistency}
\begin{lemma}[Guarantees on the virtual instance]
\label{lem:AdWords-virtual-guarantees}
The algorithm is $R$-robust and $C_k(R)$-consistent for the virtual instance.
\end{lemma}

\begin{proof}
For robustness, if $R=0$, the claim is immediate because
$\ALG_{\mathrm{virt}}\ge 0=R\OPT_{\mathrm{virt}}$. Assume from now on that
$R>0$. By \Cref{lem:AdWords-certificate} and \Cref{obs:AdWords-load},
\[
    \alpha_i+\frac{b_{ih}}{\widehat B_h}\beta_h
    \ge
    R\,b_{ih}
\]
for every virtual edge $(i,h)$. Since $\alpha_i\ge0$ and
$\beta_h\ge \widehat B_hP_h^{k-1}\ge0$, the scaled assignment
$(\alpha/R,\beta/R)$ is feasible for $(\DP^{\mathrm{Ad}})$. By weak duality and
\Cref{lem:AdWords-dual-objective},
\[
    \frac1R \ALG_{\mathrm{virt}} \ge \OPT_{\mathrm{virt}},
\]
so $\ALG_{\mathrm{virt}}\ge R\OPT_{\mathrm{virt}}$.

For consistency, let $\hat x$ be the lifted prediction: in stages
$1,\ldots,k-1$ it is the revealed prediction lifted to the corresponding stage
copies, and in stage $k$ it is an optimal residual final-stage allocation
lifted to the residual-capacity copies.
We claim that every virtual edge $(i,h)$ in the support of $\hat x$
satisfies
\[
    \alpha_i+\frac{b_{ih}}{\widehat B_h}\beta_h
    \ge
    b_{ih} \cdot C_k(R) \ .
\]
For a non-final predicted edge, $h=(j,\tau)$ and the edge lies in stage
$\tau<k$. The three cases of \Cref{lem:AdWords-certificate}, together
with \Cref{obs:AdWords-load}, give the claim: if $X_h^\tau<1$, use
\[
    P_h^{k-1}+1-f_h^\tau(x_h^\tau)\ge C_k(R);
\]
if $q_h=\tau$, use
\[
    P_h^{\tau-1}+x_h^\tau
    =
    P_h^{\tau-1}+1-X_h^{\tau-1}
    \ge C_k(R);
\]
and if $q_h<\tau$, use
\[
    P_h^{q_h}
    =
    P_h^{q_h}+1-X_h^{q_h}
    \ge C_k(R).
\]
For a final-stage predicted edge, the lifted prediction uses only the
residual-capacity advertisers $h=(j,k)$. The final-stage bound in
\Cref{lem:AdWords-certificate} plus \Cref{obs:AdWords-load} gives
\[
    \alpha_i+\frac{b_{ih}}{\widehat B_h}\beta_h
    \ge
    b_{ih}(P_h^{k-1}+1-X_h^{k-1})
    \ge
    C_k(R) \cdot b_{ih}.
\]

Using feasibility of the lifted prediction and nonnegativity of the dual
variables,

\begin{align*}
\ALG_{\mathrm{virt}}
&=
\sum_{i\in D}\alpha_i+\sum_{h\in\mathcal H}\beta_h \\
&\ge
\sum_{(i,h)\in E^{\mathrm{virt}}}
\hat x_{ih}
\left(
    \alpha_i+\frac{b_{ih}}{\widehat B_h}\beta_h
\right) \\
&\ge
C_k(R)\sum_{(i,h)\in E^{\mathrm{virt}}}b_{ih}\hat x_{ih}
=
C_k(R) \cdot \advice_{\mathrm{virt}}.
\end{align*}
This proves the claimed consistency bound for the virtual instance.
\end{proof}

\subsection{Lazy Virtual Advertiser Implementation}

\begin{lemma}[Lazy implementation]
\label{lem:lazy-virtual-implementation}
The virtual-advertiser construction can be implemented online. In stage $s$, the algorithm needs to know only the predicted allocations $\hx_j^1,\ldots,\hx_j^{s}$ for each $j \in S$.
\end{lemma}

\begin{proof}
Fix an original advertiser $j\in S$. Before the beginning of stage $s<k$, assume
that the predicted values $\hx_j^1,\ldots,\hx_j^{s-1}$ have already been revealed.
The algorithm maintains the revealed virtual advertisers
$(j,1),\ldots,(j,s-1)$, together with a residual virtual advertiser
$(j,\bot_s)$ of budget
\[
    \hB_{j,\bot_s}
    :=
    \left(1-\sum_{\tau=1}^{s-1}\hx_j^\tau\right)B_j .
\]
This residual advertiser represents all future advertisers
$(j,s),(j,s+1),\ldots,(j,k)$.

When the stage-$s$ prediction $\hx^{s}$ is revealed, compute
\[
    \hx_j^s
    :=
    \frac{1}{B_j}\sum_{i\in D_s}b_{ij}\hx_{ij}^s .
\]
We split $(j,\bot_s)$ into the stage-$s$ advertiser $(j,s)$, with budget
\[
    \hB_{j,s}:=\hx_j^s B_j,
\]
and a new residual advertiser $(j,\bot_{s+1})$, with budget
\[
    \hB_{j,\bot_{s+1}}
    :=
    \left(1-\sum_{\tau=1}^{s}\hx_j^\tau\right)B_j .
\]
Zero-budget virtual advertisers are omitted. For every earlier stage $r<s$ and every edge
$(i,(j,\bot_s))$, we split the old internal allocation proportionally by defining
\[
    x^r_{i,(j,s)}
    :=
    \frac{\widehat B_{j,s}}{\widehat B_{j,\bot_s}}
    x^r_{i,(j,\bot_s)} 
\]
and
\[
    x^r_{i,(j,\bot_{s+1})}
    :=
    \frac{\widehat B_{j,\bot_{s+1}}}{\widehat B_{j,\bot_s}}
    x^r_{i,(j,\bot_s)}.
\]
This preserves the original aggregated allocation to advertiser $j$ and gives both children the
same normalized loads, hence the same values of $X^{s-1}$ and $P^{s-1}$.

It remains to justify that aggregating all unrevealed future advertisers into
one residual advertiser is exact. Fix a collection $T$ of virtual advertisers
corresponding to the same original advertiser $j$. Assume that all
$h\in T$ have the same normalized history and use the same penalty function
$f$ in the current stage. Let
$\hB_T:=\sum_{h\in T}\hB_h$ and
$g(z):=z-\int_0^z f(t)\,dt$. 
Since $f$ is non-decreasing, $g$ is concave.

First consider any feasible allocation to the separate advertisers $h\in T$.
For each $h\in T$, let
\[
    z_h:=\frac{1}{\hB_h}\sum_{i\in D_s}b_{ij}x_{ih}.
\]
Aggregating these variables to $x_{iT}:=\sum_{h\in T}x_{ih}$ preserves the demand constraints.
The normalized load of the aggregate advertiser is
\[
    z_T
    :=
    \frac{1}{\hB_T}\sum_{i\in D_s}b_{ij}x_{iT}
    =
    \frac{1}{\hB_T}\sum_{h\in T}\hB_h z_h .
\]
Since all advertisers in $T$ have the same residual normalized capacity and
each $z_h$ is feasible, the weighted average $z_T$ is feasible as well.
Moreover, by Jensen's inequality,
\[
    \sum_{h\in T}\hB_h g(z_h)
    \le
    \hB_T
    g\left(
        \frac{\sum_{h\in T}\hB_h z_h}{\hB_T}
    \right)
    =
    \hB_T g(z_T).
\]
Thus every feasible allocation to the separate future advertisers can be
mapped to a feasible allocation to the aggregate residual advertiser with
at least the same objective value.

Conversely, consider any feasible allocation $x_{iT}$ to the aggregate
residual advertiser. Split it proportionally among the advertisers in $T$, that is, $x_{ih}:=\frac{\hB_h}{\hB_T}x_{iT}$ for all $h\in T$ and $i\in D_s$.
This preserves the demand constraints because $\sum_{h\in T}x_{ih}=x_{iT}$.
Furthermore, every $h\in T$ receives the same normalized load
\[
    \frac{1}{\hB_h}\sum_{i\in D_s}b_{ij}x_{ih}
    =
    \frac{1}{\hB_T}\sum_{i\in D_s}b_{ij}x_{iT}
    =
    z_T.
\]
Hence all individual budget constraints are preserved. The objective value is
also preserved exactly, that is, $\sum_{h\in T}\hB_h g(z_T) = \hB_T g(z_T)$.

Therefore, replacing the advertisers in $T$ by a single residual advertiser of
budget $\hB_T$ does not change the optimal value of the current-stage convex
program. Moreover, an optimal aggregate solution can always be expanded into a
symmetric optimal solution for the separate future advertisers.

Applying this argument inductively over the stages shows that the residual
advertiser is an exact online implementation of the full virtual construction:
before stage $s$, it represents precisely the unrevealed future advertisers
$(j,s),\ldots,(j,k)$. In stage $s$, the newly revealed advertiser $(j,s)$ is
split off and uses the predicted penalty, while the remaining future
advertisers are still indistinguishable and continue to be represented by
$(j,\bot_{s+1})$. After stage $k-1$, the remaining residual advertiser is
identified with $(j,k)$.
\end{proof}

\subsection{Proof of \Cref{thm:AdWords}}

By \Cref{lem:lazy-virtual-implementation}, the virtual-advertiser algorithm
can be implemented online. By \Cref{lem:AdWords-virtual-guarantees}, it is
$R$-robust and $C_k(R)$-consistent on the virtual instance. By
\Cref{lem:virtual-preserves-benchmarks}, $\OPT=\OPT_{\mathrm{virt}}$ and
$\advice=\advice_{\mathrm{virt}}$, and aggregation preserves the algorithm's value
and feasibility. Hence the same guarantees hold in the original AdWords
instance.

The  online statement is obtained in the same way, with the
$k$-stage penalties replaced by the exponential online penalties from
\Cref{sec:online}. The load inequalities in \Cref{app:online} are per
advertiser and depend only on the normalized load sequence and the time at which
the virtual advertiser is predicted; hence they apply to each virtual advertiser
created by the fractional prediction. The AdWords KKT and dual-fitting steps are
identical to \Cref{lem:AdWords:kkt,lem:AdWords-certificate,lem:AdWords-virtual-guarantees},
with the online load certificates substituted for the $k$-stage ones. Finally,
the lazy construction remains valid because all unrevealed future copies use the
same non-predicted online penalty and have the same normalized history until they
are split off. This proves the theorem for online fractional AdWords. The
vertex-weighted online bipartite matching statements with fractional predictions follow by the
specialization $b_{ij}=B_j=w_j$, or equivalently by the capacity-copy
interpretation described in \Cref{app:AdWords}.

\end{document}